\documentclass[11pt,letterpaper]{article}

\usepackage{typearea}
\typearea{14}

\usepackage{comment,algorithm,algorithmic}


\usepackage{color,colortbl}
\usepackage{float}
\usepackage{amsthm}
\usepackage{amsmath}
\usepackage{amssymb}
\usepackage{graphicx}
\usepackage{xspace}
\usepackage{nicefrac}
\usepackage[colorinlistoftodos,prependcaption,textsize=tiny]{todonotes}
\usepackage{wrapfig}

\usepackage[compact]{titlesec}
\usepackage{thmtools}
\usepackage{thm-restate}
\declaretheorem[name=Lemma]{lemma}

\definecolor{Darkblue}{rgb}{0,0,0.4}
\definecolor{Brown}{cmyk}{0,0.61,1.,0.60}
\definecolor{Purple}{cmyk}{0.45,0.86,0,0}

\usepackage[colorlinks,linkcolor=blue,filecolor=blue,citecolor=blue,urlcolor=Darkblue,pagebackref]{hyperref}
\usepackage[nameinlink]{cleveref}

\newtheorem{theorem}{Theorem}
\newtheorem{corollary}{Corollary}

\newtheorem{claim}{Claim}

\newtheorem{definition}{Definition}
\newtheorem{observation}[lemma]{Observation}

\newcommand{\namedref}[2]{\hyperref[#2]{#1~\ref*{#2}}}
\newcommand{\tableref}[1]{\namedref{Table}{#1}}

\newcommand{\propertyref}[1]{\namedref{Property}{#1}}

\newcommand{\obsref}[1]{\namedref{Observation}{#1}}

\newcommand{\conref}[1]{\hyperref[#1]{Condition~(\ref*{#1})}}

\newcommand{\E}{{\mathbb{E}}}

\newcommand{\R}{\mathbb{R}}
\newcommand{\Z}{\mathbb{Z}}

\newcommand{\poly}{{\rm poly}}
\newcommand{\fpath}{{\rm path}}
\newcommand{\froot}{{\rm root}}
\newcommand{\odd}{{\rm odd}}
\newcommand{\even}{{\rm even}}

\newcommand{\cdim}{\textsf{cdim}}
\newcommand{\SPD}{\textsf{SPD}\xspace}
\newcommand{\SPDs}{{\SPD}s\xspace}
\newcommand{\SPDdepth}{\textsf{SPDdepth}\xspace}

\def\inline#1:{\par\vskip 7pt\noindent{\bf #1:}\hskip 10pt}

\def\inline#1:{\par\vskip 7pt\noindent{\bf #1:}\hskip 10pt}

\def\blackslug{\hbox{\hskip 1pt \vrule width 4pt height 8pt
		depth 1.5pt \hskip 1pt}}

\def\QED{\quad\blackslug\lower 8.5pt\null\par}

\newcommand{\initOneLiners}{%
    \setlength{\itemsep}{0pt}
    \setlength{\parsep }{0pt}
    \setlength{\topsep }{0pt}
}
\newenvironment{OneLiners}[1][\ensuremath{\bullet}]
    {\begin{list}
        {#1}
        {\initOneLiners}}
    {\end{list}}

\newcommand{\alert}[1]{\textbf{\color{red}
		[[[#1]]]}\marginpar{\textbf{\color{red}**}}\typeout{ALERT:
		\the\inputlineno: #1}}
\definecolor{purple}{rgb}{0.294, 0, 0.71}

\floatstyle{ruled}
\newfloat{algorithm}{tbp}{loa}
\providecommand{\algorithmname}{Algorithm}
\floatname{algorithm}{\protect\algorithmname}

\usepackage{authblk}
\usepackage{pdfsync}

\begin{document}
\author[1]{Ittai Abraham}
\author[2]{Arnold Filtser}
\author[3]{Anupam Gupta}
\author[4]{Ofer Neiman}
\affil[1]{VMWare. Email: \texttt{iabraham@vmware.com}}
\affil[2]{Bar-Ilan University. Email: \texttt{arnold.filtser@biu.ac.il} }
\affil[3]{Carnegie Mellon University. Email: \texttt{anupamg@cs.cmu.edu}}
\affil[4]{Ben-Gurion University of the Negev. Email: \texttt{neimano@cs.bgu.ac.il}}

\begin{titlepage}
  \title{Metric Embedding via
    Shortest Path Decompositions\thanks{Preliminary version of this paper was published in the proceedings of STOC’18. A full version of this paper is also available at \href{https://arxiv.org/abs/1708.04073}{arxiv:1708.04073}. }}
  \maketitle

\begin{abstract}
We study the problem of embedding shortest-path metrics of weighted graphs into $\ell_p$ spaces.  We introduce a new embedding technique based on low-depth decompositions of a graph via shortest paths. The notion of Shortest Path Decomposition depth is inductively defined: A (weighed) path graph has shortest path decomposition (SPD) depth $1$. General graph has an SPD of depth $k$ if it contains a shortest path whose deletion leads to a graph, each of whose components has SPD depth at most $k-1$.  
In this paper we give an $O(k^{\min\{\nicefrac{1}{p},\nicefrac{1}{2}\}})$-distortion embedding for graphs of SPD depth at most $k$. This result is asymptotically tight for any fixed $p>1$, while for $p=1$ it is tight up to second order terms.

As a corollary of this result, we show that graphs having pathwidth $k$ embed into $\ell_p$ with distortion $O(k^{\min\{\nicefrac{1}{p},\nicefrac{1}{2}\}})$.  For $p=1$, this improves over the best previous bound of Lee and Sidiropoulos that was exponential in $k$; moreover, for other values of $p$ it gives the first embeddings whose distortion is independent of the graph size $n$. 
Furthermore, we use the fact that planar graphs have SPD depth $O(\log n)$ to give a new proof that any planar graph embeds into $\ell_1$ with distortion $O(\sqrt{\log n})$. Our approach also gives new results for graphs with bounded treewidth, and for graphs excluding a fixed minor.

\end{abstract}
\thispagestyle{empty}	
\end{titlepage}
\section{Introduction}
\label{sec:intro}

Low-distortion metric embeddings are a crucial component in the modern
algorithmist toolkit. Indeed, they have applications in approximation
algorithms \cite{LLR95}, online algorithms \cite{BBMN15}, distributed algorithms \cite{KKMPT12}, and for solving
linear systems and computing graph sparsifiers \cite{ST04}. Given a (finite) metric space $(V,d)$, a map $\phi: V \to
\R^D$, and a norm $\|\cdot\|$, the \emph{contraction} and
\emph{expansion} of the map $\phi$ are the smallest $\tau, \rho \geq 1$,
respectively, such that for every pair $x,y\in V$,
\[ \frac1\tau \leq \frac{\| \phi(x) - \phi(y) \|}{d(x,y)} \leq
\rho~~. \] The \emph{distortion} of the map is then $\tau \cdot \rho$.

In
this paper we will investigate embeddings into $\ell_p$ norms; the most
prominent of which are the Euclidean norm $\ell_2$ and the cut norm
$\ell_1$; the former for obvious reasons, and the latter because of its
close connection to graph partitioning problems, and in particular the
Sparsest Cut problem. Specifically, the ratio between the Sparsest Cut and the multicommodity flow equals the distortion of the optimal embedding into $\ell_1$ (note that every $\ell_1$ metric is a linear combination of cut metrics. See \cite{LLR95,GNRS04} for more details on the connection to multicommodity flows).

We focus on embedding of metrics arising from certain graph families.
Indeed, since general $n$-point metrics require $\Omega(\nicefrac{\log
	n}{p})$-distortion to embed into $\ell_p$-norms, much attention was
given to embeddings of restricted graph families that arise in
practice. (Embedding an (edge-weighted) graph is short-hand for
embedding the shortest path metric of the graph generated by these
edge-weights.) Since the class of graphs embeddable with some distortion
into some target normed space is closed under taking minors, it is
natural to focus on minor-closed graph families. A long-standing
open problem in this area to decide whether all non-trivial minor-closed families of
graphs embed into $\ell_1$ with distortion depending only on the graph
family, and not the size $n$ of the graph.

While this question remains unresolved in general, there has been some
progress on special classes of graphs.
The class of outerplanar
graphs (which exclude $K_{2,3}$ and $K_4$ as a minor) embeds
isometrically into $\ell_1$; this follows from results of Okamura and
Seymour \cite{OS81}. Following \cite{GNRS04}, Chakrabarti et al.~\cite{CJLV08} show
that every graph with treewidth-$2$ (which excludes $K_4$ as a minor)
embeds into $\ell_1$ with distortion $2$ (which is tight, as shown by \cite{LeeR10}).
Lee and Sidiropoulos \cite{LS13} showed that every graph with pathwidth
$k$ can be embedded into $\ell_1$ with distortion $(4k)^{k^3+1}$.
See \Cref{sec:related} for additional
results.

We note that $\ell_2$ is a potentially more natural and useful target space than $\ell_1$ (in particular, finite subsets of $\ell_2$ embed  isometrically into $\ell_1$). Alas, there are only few (natural) families of metrics that admit constant distortion embedding into Euclidean space, such as ``snowflakes'' of
doubling metrics \cite{A83},
doubling trees \cite{GKL03} and graphs
of bounded bandwidth \cite{BCMN13}. All these families have bounded
doubling dimension. (For definitions, see	 \Cref{sec:preliminaries}.)

\subsection{Our Results}

In this paper we develop a new technique for embedding graphs
into $\ell_p$ spaces with small distortion. 
We introduce the notion of Shortest Path Decomposition (SPD) of bounded depth.
Every (weighted) path graph has an \SPD of depth $1$. A graph $G$ has an \SPD of depth
$k$ if there exists a \emph{shortest path} $P$, such that deleting $P$ from the graph $G$ (that is, deleting all the vertices on $P$ and the adjacent edges)
results in a graph whose connected components all have \SPD of depth at most $k-1$. (An alternative
definition appears in \Cref{def:SPD}.) Our main result is
the following.

\begin{theorem}[Embeddings for \SPD Families]
	\label{thm:main}
	Let $G=\left(V,E\right)$ be a weighted graph with
	an \SPD of depth $k$. Then there exists an embedding
	$f:V\to\ell_{p}$ with distortion $O(k^{\nicefrac{1}{p}})$.
\end{theorem}

{\em Remark:} Since finite subsets of $\ell_2$ embed isometrically into $\ell_p$ for any $1\le p\le\infty$, we get that the distortion of \Cref{thm:main} is never larger than $O(\sqrt{k})$.

\paragraph{Graph families with \SPD of small depth.}

We will show that graphs of pathwidth $k$ have \SPD of depth $k+1$, and thus obtain the following result as a simple corollary of \Cref{thm:main}.

\begin{theorem}[Pathwidth Theorem]
	\label{thm:path-main}
	Any graph with pathwidth $k$ embeds into $\ell_p$ with distortion
	$O(k^{\min\{\nicefrac{1}{p},\nicefrac{1}{2}\}})$.
\end{theorem}

Note that this is a super-exponential
improvement over the best previous
distortion bound of $O(k)^{k^3}$, by Lee and Sidiropoulos \cite{LS13}.  Their
approach was based on probabilistic embedding into trees, which implies
embedding only into $\ell_1$. Such an approach cannot yield distortion
better than $O(k)$, due to known lower bounds for the diamond graph
\cite{GNRS04}, that has pathwidth $k+1$.
Our embedding holds for any $\ell_p$ space, and we can
overcome the barrier of $\Theta(k)$. In particular, we obtain embeddings of pathwidth-$k$ graphs into both $\ell_2$ and $\ell_1$
with distortion $O(\sqrt{k})$. Moreover, an embedding with this
distortion can be found efficiently via semidefinite-programming; see,
e.g.,~\cite{LLR95}, even without access to the actual path decomposition (which is NP-hard even to approximate \cite{BGHK92}).
We remark that graphs of bounded pathwidth can have arbitrarily large doubling dimension (exhibited by star graphs that have pathwidth~1), and thus our result is a noteworthy example of a non-trivial Euclidean embedding with constant distortion for a family of metrics with unbounded doubling dimension.

Since graphs of treewidth $k$ have pathwidth $O(k\log n)$ (see, e.g., \cite{KS93}), \Cref{thm:path-main} provides an embedding of such graphs into $\ell_p$ with distortion $O((k\log n)^{1/p})$.
This strictly improves the best previously known bound, which follows from a theorem in \cite{KLMN04}
(who obtained distortion $O(k^{1-1/p}\log^{1/p}n)$ ), for any $p>2$, and matches it for $1\le p\le 2$. While \cite{KK16} obtained recently a distortion bound with improved dependence on $k$, their result $O((\log (k\log n))^{1-1/p}(\log^{1/p} n))$ has sub-optimal dependence on $n$.

Moreover, we derive several other results for planar graphs, and more generally graphs excluding a fixed minor. Even though these families have unbounded pathwidth, we show that they have \SPD of depth $O(\log n)$. These results
are summarized in \tableref{tab:resultSummery}, they either
improve on the state-of-the-art, or provide matching bounds using a
new approach.

In \Cref{sec:dimension} we show that we can slightly modify our construction of \Cref{thm:main} so that the dimension of the host space will be $O(k\log n)$, while maintaining the same distortion guarantee. This implies that graphs excluding $H$ as a minor admit an embedding into $\ell_\infty^{O(g(H)\cdot\log^2n)}$ with constant distortion (this constant is independent of $H$). See \Cref{thm:dimension} and the discussion therein.

\begin{table*}[]
	\centering
		\label{tab:resultSummery}
		\begin{tabular}{|l|l|lr|}
			\hline   \textbf{Graph Family}& \textbf{Our results.}  $\qquad$& \textbf{Previous
				results} & \\ \hline \hline
			Pathwidth $k$ & $O(k^{\nicefrac{1}{p}})$ & $(4k)^{k^3+1}$ into
			$\ell_1$ & \cite{LS13} \\ \hline
			Treewidth $k$ & $O((k \log n)^{\nicefrac{1}{p}})$ & $O(k^{1-1/p} \cdot
			\log^{1/p} n)$ & \cite{KLMN04} \\
			& & $O((\log (k\log n))^{1-1/p}(\log^{1/p} n))$ & \cite{KK16} \\ \hline
			Planar & $O(\log ^{\nicefrac{1}{p}}n)$  &
			$O(\log^{\nicefrac{1}{p}}n)$ & \cite{Rao99} \\ \hline
			$H$-minor-free & $O((g(H)\log n)^{\nicefrac{1}{p}})$ &
			$O(|H|^{1-\nicefrac{1}{p}}\log^{\nicefrac{1}{p}}n)$ & \cite{AGGNT19}+\cite{KLMN04} \\ \hline
            $H$-minor-free & $O(1)$ into $\ell_{\infty}^{O(g(H)\cdot\log^2n)}$ &
			$O(|H|^2)$ into $\ell_\infty^{O(3^{|H|}\log|H|\log n)}$& \cite{KLMN04}   \\ \hline
		\end{tabular}
		\caption{\small Our and previous results for embedding certain graph families into $\ell_p$. (For $H$-minor-free graphs, $g(H)$ is some function of $|H|$.)}
\end{table*}

Our result of \Cref{thm:path-main} (and thus also
\Cref{thm:main}) is asymptotically tight for any fixed $p >1$. 
\footnote{A previous version of this paper contained a lower bound for the case $p=1$. The proof of this lower bound was wrong, and it is removed from the current version.}
The
family exhibiting this fact is the diamond graphs. 
\begin{theorem}[\cite{NR02,LN04,JS09}]
	\label{thm:Diamond}
	For any fixed $p >1$ and every $k\ge 1$, there exists a graph
	$G=(V,E)$ with pathwidth-$k$, such that every embedding
	$f:V\to\ell_{p}$ has distortion
	$\Omega(k^{\min\{\nicefrac{1}{p},\nicefrac{1}{2}\}})$.
\end{theorem}
The bound in \Cref{thm:Diamond} was proven first for $p=2$ in \cite{NR02}, generalized to  $1<p\le 2$ in \cite{LN04} and for $p\ge2$ by \cite{JS09} (see also \cite{MN13,JLM11}).

\subsection{Technical Ideas}\label{sec:technical}
Many known embeddings ~\cite{B85, Rao99, KLMN04, ABN06} are based on a collection of
1-dimensional embeddings, where we embed each point to its distance
from a given subset of points.
We follow this approach, but differ in two aspects.  Firstly, the subset
of points we use is not based on random sampling \cite{B85} or probabilistic
clustering \cite{Rao99}. Rather, inspired by the works of~\cite{And86}
and~\cite{AGGNT19}, the subset used is a geodesic shortest path. The
second is that our embedding is not 1-dimensional but 2-dimensional:
this seemingly small change crucially allows us to use the structure of
the shortest paths to our advantage.

The \SPD induces a collection of shortest paths (each shortest path lies in some connected component). A natural initial attempt is to
embed a vertex $v$  relative to a geodesic path $P$ using two
dimensions:\footnote{In fact, we use different dimensions for each connected component.}

\begin{itemize}
	\item The \textit{first} coordinate $\Delta_1$ is the distance to the path $d(v,P)$.
	
	\item The \textit{second} coordinate $\Delta_2$ is the distance $d(v,r)$ to the
	endpoint of the path, called its ``root''.
\end{itemize}

\begin{figure}[H]
	\centering{\includegraphics[scale=1]{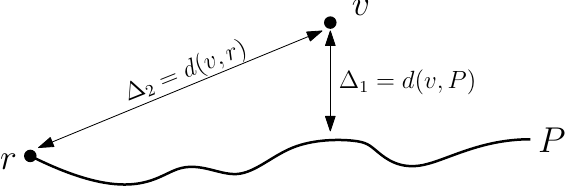}}
	\caption{\label{fig:intuition}\small An illustration of our initial
		attempt. The first coordinate $\Delta_1$ is the distance to the path
		$d(v,P)$. The second coordinate $\Delta_2$ is the distance $d(v,r)$
		to the endpoint of the path, called its ``root''.}
\end{figure}
Unfortunately, this embedding may have unbounded expansion: If two vertices $u,v$ are separated by some shortest path, in future iterations $v$ may
have a large distance to the root of a path $P$ in its component, while
$u$ has zero in that coordinate (because it's not in that component), incurring a large stretch.
The natural fix is to enforce a Lipschitz condition on every coordinate: for $v$ in cluster $X$, we \textit{truncate} the value $v$ can receive at $O(d_G(v,V\setminus X))$. I.e., a vertex close to the boundary
of $X$ cannot get a large value. Using the fact that the \SPD has depth $k$, each vertex will have only $O(k)$ nonzero coordinates, which implies expansion $O(k^{1/p})$.

To bound the contraction, for each pair $u,v$ we consider the {\em
	first} path $P$ in the \SPD that lies ``close'' to $\{u,v\}$ or
separates them to different connected components. Then we show that at
least one of the two coordinates should give sufficient contribution.

But what about the effect of truncation on contraction? A careful
recursive argument shows that the contribution to $u,v$ from the first
coordinate (the distance from the path $P$) is essentially not affected
by this truncation. Hence the argument in cases (a) and (b) of
\Cref{fig:Cases} still works. However, the argument using the distance
to the root of $P$, case~(c), can be ruined. Solving this  issue requires some new non-trivial ideas. Our solution is to introduce a probabilistic sawtooth function that replaces the simple truncation. The main technical part of the paper is devoted to showing that a collection of these functions for all possible distance scales, with appropriate random shifts, suffices to control the expected contraction in case (c), for {\em all relevant pairs simultaneously}.

\begin{figure*}[]
	\centering{\includegraphics[width=1.0\textwidth]{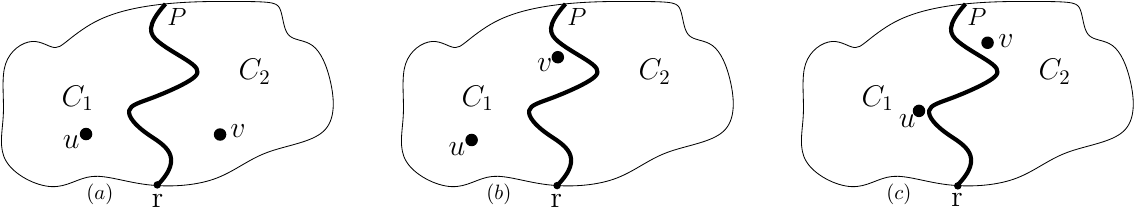}}
	\caption{\label{fig:Cases}\small 
		A shortest path $P$, rooted at $r$, partitions cluster $X$ into clusters $C_1$ and $C_2$.  In cases (a) and
		(b), the first coordinate (the distance to $P$) provides sufficient
		contribution. In case (c) the second coordinate (the distance to
		root $r$)
		provides the contribution. }
\end{figure*}

\subsection{Other Related Work}
\label{sec:related}

There has been work on embedding several other graph families into normed spaces:
Chekuri et al. \cite{CGNRS01} extend the Okamura and Seymour bound for outerplanar graphs to $k$-outerplanar graphs, and showed that these embed into $\ell_1$ with distortion
$2^{O(k)}$.
Rao \cite{Rao99} (see also \cite{KLMN04}) embed planar graphs into
$\ell_p$ with distortion $O(\log^{1/p} n)$. For graphs with genus $g$, \cite{LS10} showed an embedding into Euclidean space with distortion $O(\log g+\sqrt{\log n})$.
Finally, for $H$-minor-free graphs, combining the results
of \cite{AGGNT19, KLMN04} give $\ell_p$-embeddings with
$O(|H|^{1-1/p}\log^{1/p} n)$ distortion.

Following ~\cite{And86,Mil86}, the idea of using geodesic shortest paths
to decompose the graph has been used for many algorithmic tasks: MPLS
routing \cite{GKmpls01}, directed connectivity, distance labels and
compact routing \cite{T04}, object location \cite{AG06}, and nearest neighbor search \cite{ACKW15}. 

Given a tree $T$, Matou{\v{s}}ek \cite{M99} recursively defined the \emph{caterpillar dimension} $\cdim(T)$ as follows: the $\cdim$ of a singleton vertex is $0$. The $\cdim$ of a tree is $k$ if there are a set of paths $P_1,\dots,P_s$ intersecting in a single vertex, such that removing the edges in all these paths (as opposed to vertices in our definition) results in connected components each with  $\cdim$ at most $k-1$.
Matou{\v{s}}ek showed that every tree $T$ with $\cdim(T)$ embeds into every $\ell_p$ space with distortion $O_p(\log(\cdim(T)))^{\min\{\frac12,\frac1p\}}$. 

In a follow up paper, \cite{Fil20faces} (the second author) generalized our definition of \SPD to partial-\SPD (allowing the lower level in the partition hierarchy to be general subgraph rather than only a shortest path).
Given a weighted planar graph $G=(V,E,w)$ with a subset of terminals $K$, a face cover is a subset of faces such that every terminal lies on some face from the cover.
Given a face cover of size $\gamma$, using our embedding result for \SPD, \cite{Fil20faces} shows that the terminal set $K$ can be embedded into $\ell_1$ with distortion $O(\sqrt{\log \gamma})$.

In another follow-up paper, the second author \cite{Fil20Scattering}, created strong sparse partitions for graphs with bounded \SPDdepth. These were later used to obtain constant distortion solution for the Steiner point removal problem, and also creating approximation algorithms for the universal Steiner tree problem, and universal Traveling salesman problem \cite{JLNRS05}.

\section{Preliminaries and Notation}
\label{sec:preliminaries}

For $k \in \Z$, let $[k]:=\{1,\dots,k\}$. For $p \geq 1$, the
$\ell_p$-norm of a vector $x=(x_1,\dots,x_d)\in \R^d$ is $\Vert
x \Vert_{p} : =(\sum_{i=1}^{d}|x_{i}|^{p})^{1/{p}}$, where $\Vert
x\Vert_{\infty} :=\max_{i}|x_{i}|$.

\emph{Doubling dimension.}  The doubling dimension of a metric is a
measure of its local ``growth rate''. Formally, a metric space $(X,d)$ has
\emph{doubling dimension} $\lambda_X$ if for every $x\in X$ and radius
$r$, the ball $B(x,r)$ can be covered by $2^{\lambda_X}$ balls of radius
$\frac{r}{2}$. A family is \emph{doubling} if the doubling dimension of all metrics in it is bounded by some universal constant.

\emph{Graphs.} We consider connected undirected graphs $G=(V,E)$ with edge weights
$w: E \to \R_{> 0}$. Let $d_{G}$ denote the shortest path metric in
$G$; we drop subscripts when there is no ambiguity. For a vertex $x\in
V$ and a set $A\subseteq V$, let $d_{G}(x,A):=\min_{a\in A}d(x,a)$,
where $d_{G}(x,\emptyset):= \infty$.  For a subset of vertices
$A\subseteq V$, let $G[A]$ denote the induced graph on $A$, and let
$d_{A}:=d_{G[A]}$ be the shortest path metric in the induced graph. Let
$G\setminus A := G[V\setminus A]$ be the graph after deleting the vertex
set $A$ from $G$.

\emph{Special graph families.}
Given a graph $G=(V,E)$, a \emph{tree decomposition} of $G$ is a tree
$T$ with nodes $B_1,\dots,B_s$ (called \emph{bags}) where each $B_i$ is
a subset of $V$ such that the following properties hold:
\begin{OneLiners}
	\item For every edge $\{u,v\}\in E$, there is a bag $B_i$ containing
	both $u$ and $v$.
	\item For every vertex $v\in V$, the set of bags containing $v$ form a
	connected subtree of $T$.
\end{OneLiners}
The \emph{width} of a tree decomposition is $\max_i\{|B_i|-1\}$. The \emph{treewidth} of $G$ is the minimal
width of a tree decomposition of $G$.

A \emph{path decomposition} of $G$ is a special kind of tree
decomposition where the underlying tree is a path. The \emph{pathwidth}
of $G$ is the minimal width of a path decomposition of $G$.

A graph $H$ is a \emph{minor} of a graph $G$ if we can obtain $H$ from
$G$ by edge deletions/contractions, and vertex deletions.  A graph
family $\mathcal{G}$ is \emph{$H$-minor-free} if no graph
$G\in\mathcal{G}$ has $H$ as a minor.

\subsection{The Sawtooth Function}

An important component in our embeddings will be the following sawtooth
function. For $t\in\mathbb{N}$, we define $g_{t}:\R_{+}\rightarrow\R$
the \emph{sawtooth} function w.r.t.\ $2^{t}$ as follows. For $x\ge0$, if
$q_{x} := \lfloor x/2^{t+1} \rfloor$ then
\begin{gather*}
g_{t}(x) :=2^{t}-\left|x-\left(q_{x}\cdot2^{t+1}+2^{t}\right)\right|.
\end{gather*}
\Cref{fig:g_tGraph} can help visualize this function. The following observation is straightforward.

\begin{figure}[t]
	\centering{\includegraphics[scale=0.62]{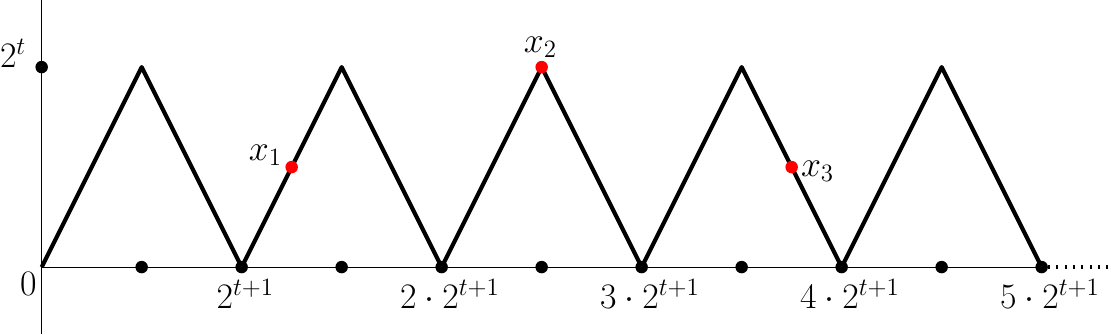}}
	\caption{\label{fig:g_tGraph} \small  The graph of the ``sawtooth''
		function $g_{t}$. The points $x_1=5\cdot 2^{t-1}$ and
		$x_3=15\cdot 2^{t-1}$ are mapped to $2^{t-1}$, while
		$x_2=10\cdot 2^{t-1}$ is mapped to $2^t$.}
\end{figure}

\begin{observation}\label{Obs:Sawtooth}
	The Sawtooth function $g_{t}$ is $1$-Lipschitz, bounded by $2^{t}$,
	periodic with period $2^{t+1}$.
\end{observation}	

To make the proofs cleaner, we define an auxiliary function given
parameters $\alpha \in [0,1]$, $\beta\in [0,4]$:
\begin{gather}
  g_{t, \alpha, \beta}(x) := g_{t}(\beta\cdot x + \alpha\cdot 2^{t+1}). \label{eq:h-def}
\end{gather}
Note that by \Cref{Obs:Sawtooth}, $g_{t, \alpha, \beta}$ is $\beta$-Lipschitz and bounded by $2^{t}$.
The proof of the following lemma appears in \Cref{app:MissBkIdn}.
\begin{restatable}[Sawtooth Lemma]{lemma}{BrkIdn}\label{lem:BrkIdn}
	Let $x,y\in\R_{+}$. Let $\alpha\in[0,1]$, $\beta\in [0,4]$ be drawn
	uniformly and independently. The following properties hold:
	\begin{enumerate}
		\item \label{Pr:Avg} $\mathbb{E}_{\alpha,\beta}\left[\;g_{t, \alpha, \beta}(
		x)\;\right]=2^{t-1}$.
		\item \label{Pr:SmallandLargeGap}
		$\mathbb{E}_{\alpha,\beta}\left[\,\left|g_{t, \alpha, \beta}(
		x)-g_{t, \alpha, \beta}(y)\right|\,\right]= \Omega(\min\{|x-y|,
            2^t\})$.
	\end{enumerate}
\end{restatable}

\section{Shortest Path Decompositions}
\label{sec:spds}

Our embeddings will crucially depend on the notion of shortest path
decompositions.
In the introduction we provided a recursive definition for \SPD. Here we
show an equivalent definition which will be more suitable for our
purposes.
\begin{definition}[Shortest Path Decomposition\label{def:SPD} (\SPD)]
	Given a weighted graph $G=(V,E,w)$, a \SPD of depth $k$ is a pair
	$\left\{ \mathcal{X},\mathcal{P}\right\} $, where $\mathcal{X}$ is a
	collection $\mathcal{X}_{1},\dots,\mathcal{X}_{k}$ of partial partitions of
	$V$ \footnote{i.e. for every $X\in \mathcal{X}_i$, $X\subseteq V$, and
		for every different subsets $X,X'\in\mathcal{X}_{i}$, $X\cap
		X'=\emptyset$.}, and $\mathcal{P}$ is a collection of sets of paths
	$\mathcal{P}_{1},\dots,\mathcal{P}_{k}$, where $\mathcal{X}_1=\{V\}$,
	$\mathcal{X}_k=\mathcal{P}_k$, and the following properties hold:
	\begin{enumerate}
		\item For every $1\leq i\leq k$ and every subset $X\in\mathcal{X}_{i}$,
		there exist a unique path $P_X\in\mathcal{P}_{i}$ such that $P_{X}$
		is a shortest path in $G[X]$.
		\item For every $2\leq i\leq k$, $\mathcal{X}_i$ consists of all
		connected components of $G[X\setminus P_{X}]$ over all $X\in\mathcal{X}_{i-1}$.
	\end{enumerate}
\end{definition}

In other words, $\bigcup_{i=1}^k\mathcal{P}_k$ is a partition of $V$
into paths, where each path $P_X$ is a shortest path in the component
$X$ it belongs to at the point it is deleted.

For a given graph $G$ let $\SPDdepth(G)$ be the minimum $k$ such that
$G$ admits an \SPD of depth $k$. For a given family of graphs
$\mathcal{G}$ let $\SPDdepth(\mathcal{G}) := \max_{G \in
	\mathcal{G}}\{\SPDdepth(G)\}$. In the following we consider the
\SPDdepth of some graph families.

\subsection{The SPD Depth for Various Graph Families}

One advantage of defining the shortest path decomposition is that
several well-known graph families have bounded depth \SPD.

\begin{itemize}
	\item \emph{Pathwidth.} Every graph $G=(V,E,w)$ with \emph{pathwidth}
	$k$ has an \SPDdepth of $k+1$. Indeed, let $\mathcal{T}= \langle
	\mathcal{B}_1,\dots, \mathcal{B}_s \rangle$ be a path decomposition of
	$G$, where $\mathcal{B}_1,\mathcal{B}_s$ are the two bags at the end
	of this path. Choose arbitrary vertices $x\in \mathcal{B}_1$ and
	$y\in\mathcal{B}_s$, and let $P$ be a shortest path in $G$ from $x$ to
	$y$. By the definition of a path decomposition, the path $P$ contains
	at least one vertex from every bag $\mathcal{B}_i$. Hence, deleting
	the vertices of $P$ would reduce the size of each bag by one;
	consequently each connected component of $G\setminus P$ has pathwidth $k-1$, and by induction \SPDdepth $k$. Finally, a connected component of pathwidth 0 is necessarily a singleton, which has \SPDdepth $1$.

	\item \sloppy\emph{Treewidth.} Since every treewidth-$k$ graph has pathwidth $O(k\log
	n)$, treewidth-$k$ graphs have \SPDdepth $O(k\log n)$.
	
	\item \emph{Planar.} Using cycle separators~\cite{Mil86} as
	in~\cite{T04, GKmpls01}, every planar graph has \SPDdepth
	$O(\log n)$; this follows as each cycle separator can be constructed
	as union of two shortest paths.
	
	\item \emph{Minor-free.} Finally, every $H$-minor-free graph
	admits a balanced separator consisting of $g(H)$ shortest paths (for
	some function $g$)~\cite{AG06}, and hence has an \SPDdepth
	$O(g(H)\cdot \log n)$.
\end{itemize}

Combining these observation with~\Cref{thm:main}, we get the
following set of results:
\begin{corollary}
	Consider an $n$-vertex weighted graph $G$, \Cref{thm:main}
	implies the following:
	\begin{OneLiners}
		\item\label{Cor:path} If $G$ has pathwidth $k$, it embeds into
		$\ell_p$ with distortion $O(k^{\nicefrac{1}{p}})$.
		
		\item\label{Cor:tree} \sloppy If $G$ has treewidth $k$, it embeds into
		$\ell_p$ with distortion $O((k\log n)^{\nicefrac{1}{p}})$.
		
		\item\label{Cor:planar} If $G$ is planar, it embeds into $\ell_p$ with
		distortion $O(\log^{\nicefrac{1}{p}} n)$.
		
		\item\label{Cor:Minor} For every fixed $H$, if $G$ excludes $H$ as a
		minor, it embeds into $\ell_p$ with distortion $O(\log^{\nicefrac{1}{p}}
		n)$, where the constant in the big-O depends on $H$.
	\end{OneLiners}
\end{corollary}

As mentioned in \Cref{sec:intro}, we get a substantial improvement for
the pathwidth case.  Our result for treewidth improves upon that
from~\cite{KLMN04} for $p>2$; they got $O(k^{1-\nicefrac{1}{p}}(\log
n)^{\nicefrac{1}{p}})$ distortion compared to our $O((k \log
n)^{\nicefrac{1}{p}})$. Our result appears to be closer to the truth,
since the distortion tends to $O(1)$ as
$p\to\infty$.
Our results for planar graphs
match the current state-of-the-art.

Finally our results for minor-free graphs depend on the Robertson-Seymour
decomposition, and hence are currently better only for large values of
$p$. (It remains an open question to improve the \SPDdepth of
$H$-minor-free graphs to have a $\poly(|H|) \log n$ dependence, perhaps
using the ideas from~\cite{AGGNT19}.)
In general, we hope that our results will be useful in getting other
embedding results, and will spur further work on understanding shortest
path separators.

We note that there exist graphs with large \SPDdepth. For instance, the clique graph $K_{n}$ has \SPDdepth of $\frac n2$, as each shortest path contains at most $2$ vertices.
Moreover, there are sparse graphs with very large \SPDdepth. Specifically, an $n$-vertex constant degree expander has  \SPDdepth of $n^{\Omega(1)}$. Indeed, denote by $k$ the \SPDdepth of some constant degree expander $G$.
According to \Cref{thm:dimension}, $G$ can be embedded into $\ell_{\infty}^{O(k\log n)}$ with distortion $O(1)$. However, according to Rabinovich \cite{Rab08}, in order to embed a constant degree expander into $\ell_{\infty}$ with distortion $D$, $n^{\Omega(\nicefrac1D)}$ coordinates are required. It follows that $k=n^{\Omega(1)}$.

\begin{wrapfigure}{r}{0.20\textwidth}
	\begin{center}
		\vspace{-20pt}
		\includegraphics[width=0.45\textwidth]{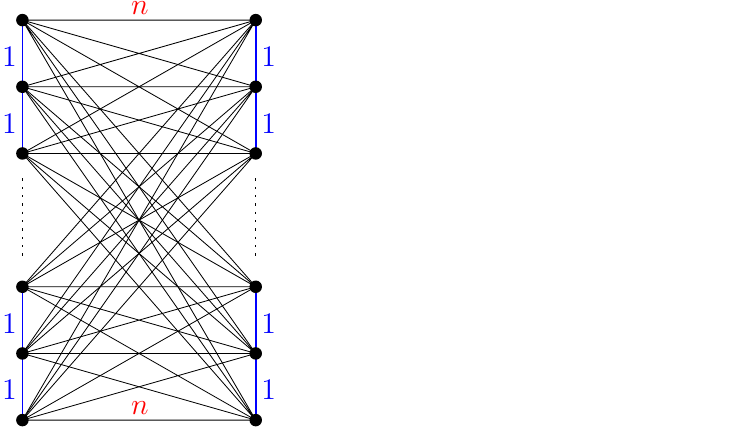}
		\vspace{-7pt}
	\end{center}
	\vspace{-20pt}
\end{wrapfigure}

On the other hand, there are graphs with \SPDdepth $2$ that contain arbitrarily large cliques. For example see the graph drawn on the right. The graph consist of two sets $\{L,R\}$ each containing $n$ vertices. The graph restricted to the vertices of $L$ (resp. $R$) consist of a shortest path with unit weight edges.
In addition, for every pair of vertices $v\in L$ and $u\in R$ we add an edge of weight $n$.
Note that $G$ contains the full bipartite graph $K_{n,n}$ as a subgraph (and in particular $K_n$ as a minor). It is straightforward that $G$ has \SPDdepth $2$.
Note also, that by subdividing each edge of weight $n$ to $n$ unit weight edges, we will get an unweighted graph of \SPDdepth $3$ that contains $K_n$ as a minor.

\section{The Embedding Algorithm}
\label{sec:embed}

Let $G=\left(V,E\right)$ be a weighted graph, and let $\left\{
\mathcal{X},\mathcal{P}\right\} =\left\{ \left\{ \mathcal{X}_{1},,\dots,\mathcal{X}_{k}\right\} ,\left\{ \mathcal{P}_{1},\dots,\mathcal{P}_{k}\right\} \right\} $ be an \SPD of depth
$k$ for $G$. By scaling, we can assume that the minimum weight of an
edge is $1$; let $M\in \mathbb{N}$ be the minimal such that the diameter
of $G$ is strictly bounded by $2^M$. Pick $\alpha\in[0,1]$ and
$\beta\in[0,4]$ uniformly and independently.

For every $i \in [k]$, and $X\in\mathcal{X}_{i}$, we now construct an
embedding $f_{X}:V\rightarrow\R^{D}$ (for some number of dimensions
$D\in\mathbb{N}$).  This map $f_{X}$ consists of two parts.

\medskip 	
\underline{\emph{First coordinate: Distance to the Path.}} The first coordinate of
the embedding implements the distance to the path $P_X$, and is denoted
by ${f}^{\fpath}_{X}$. Let $X_{1},\dots, X_{s}\in\mathcal{X}_{i+1}$ be
the connected components of $G\left[X\setminus P_{X}\right]$ (note that
it is also possible that $s=0$).  We use a separate coordinate for each
$X_{j}$, and hence $f^{\fpath}_X: V \to \R^s$. Moreover, for $v\in X$ we
truncate at $2\,d_G(v,V\setminus X)$ in order to guarantee
Lipschitz-ness.  In particular, the coordinate corresponding to $X_j$ is
set to
\[
\left({f}_{X}^{\fpath}\right)_{X_{j}}(v)=\begin{cases}
\min\left\{ d_X(v,P_{X}),2 d_{G}(v,V\setminus X)\right\} \qquad & \text{if }v\in X_{j},\\
0 & \text{otherwise}.
\end{cases}
\]
See \Cref{fig:fhatIlu} for an illustration.
\begin{figure}[h]
	\centering{\includegraphics[scale=0.8]{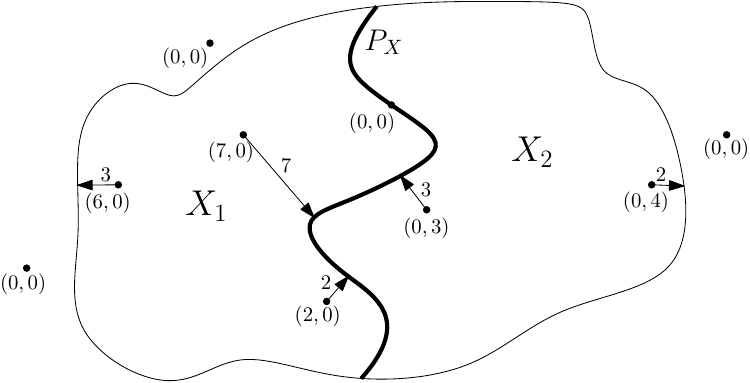}}
	\caption{\label{fig:fhatIlu}\small \emph{The set $X\in\mathcal{X}_{i}$
			surrounded by a closed curve.  The path $P_{X}$ partitions $X$
			into $X_{1},X_{2}\in\mathcal{X}_{i+1}$. The embedding
			${f}^{\fpath}_{X}$ consists of two coordinates, and represented in
			the figure by a horizontal vector next to each vertex, where the
			first entry is w.r.t.\ $X_{1}$ and the second w.r.t.\ $X_{2}$.
			Each point on $P_{X}$, or not in $X$ maps to $0$ in both the
			coordinates.  Each point in $X_{1}$ maps to $\min\left\{
			d_{X}\left(v,P_{X}\right),2 d_{G}\left(v,V\setminus
			X\right)\right\} $ in the first coordinate and to $0$ in the
			second.}}
\end{figure}

\medskip
\underline{\emph{Second coordinate: Distance to the Root.}} The second part
is denoted ${f}^{\froot}_{X}$, which is intended to capture the distance
from the root $r$ of the path.  Again, to get the Lipschitz-ness, we
would like to truncate the value at $2\, d_G(v,V\setminus X)$ as we did
for ${f}^{\fpath}_{X}$.  However, a problem with this idea is that the
root $r$ can be arbitrarily far from some pair $u,v$ that needs
contribution from this coordinate. And hence, even if
$|d_G(u,r)-d_G(v,r)|\approx d_G(u,v)$, there may be no contribution
after the truncation. So we use the sawtooth function.

Specifically, we replace the ideal contribution $d_G(v,r)$ by the
sawtooth function $g_{t}(d_G(v,r))$, where the scale $t$ for the function
is chosen such that $2^t\approx d_G(v,V\setminus X)$. To avoid the case
that two nearby points use two different scales (and hence to guarantee
Lipschitz-ness), we take an appropriate linear combination of the two
distance scales closest to $2\,d_G(v,V\setminus X)$. Recall that the
sawtooth function does not guarantee contribution for $u,v$ due to its
periodicity: we may be unlucky and have $g_{t}(d_G(v,r)) = g_{t}(d_G(u,r))$
even when $d_G(v,r)$ and $d_G(u,r)$ are very different. To guarantee a
large enough contribution for all relevant pairs simultaneously, we add
a random shift $\alpha$, and apply a random ``stretch'' $\beta$ to
$d_G(v,r)$ before feeding it to $g_{t}$. \Cref{lem:BrkIdn} then shows
that many of the choices of $\alpha$ and $\beta$ give substantially
different values for $u,v$.

Formally, the mapping is as follows. The function ${f}^{\froot}_{X}$
consists of $M+1$ coordinates, one for each distance scale $t \in \{0,
1, \ldots, M\}$. The coordinate corresponding to $t$ is denoted by
${f}^{\froot}_{X,t}$.  Let $r$ be an arbitrary endpoint of $P_{X}$; we
will call $r$ the ``root'' of $P_{X}$.
Let $t_{v}\in\mathbb{N}$ be such that $2\, d_{G}(v,V\setminus
X)\in[2^{t_{v}},2^{t_{v}+1})$. Set $\lambda_v=\frac{2\,
	d_{G}(v,V\setminus X)-2^{t_v}}{2^{t_v}}$.  Note that $0\le
\lambda_v<1$.
For $v \in X$, we define
\begin{gather}
{f}_{X,t}^{\froot}(v)=\begin{cases}
\lambda_{v}\cdot g_{t, \alpha, \beta}(d_{X}(v,r))\qquad\qquad & \text{if
}t=t_{v}+1,\\
(1-\lambda_{v})\cdot g_{t, \alpha, \beta}(d_{X}(v,r)) & \text{if }t=t_{v},\\
0 & \text{otherwise}.
\end{cases}
\label{eq:ftilde}
\end{gather}
Recall that $g_{t, \alpha, \beta}$ was defined in~(\ref{eq:h-def}).
For all nodes $v\notin X$, we set ${f}^{\froot}_{X}(v)=\vec{0}$.

Define the map $f_X = {f}^{\fpath}_{X} \oplus {f}^{\froot}_{X}$, and the final
embedding is
\[ f=\bigoplus_{i=1}^k\bigoplus_{X\in\mathcal{X}_{i}}f_{X}, \] i.e., the
concatenation of all the constructed embeddings. Before we start the
analysis, let us record some simple observations.
\begin{observation}
	\label{obv:coordinates}
	For the map $f$ defined above, the following hold:
	\begin{OneLiners}
		\item The number of coordinates in $f$ does not depend on
		$\alpha,\beta$.
		
		\item For every $X\in\mathcal{X}_{i}$ and $v\notin X$, the map
		$f_X(v)$ is the constant vector $\vec{0}$.
		\item For every $X\in\mathcal{X}_{i}$ and $v\in X$, the map $f_{X}$ is
		nonzero in at most $3$ coordinates.
	\end{OneLiners}
	Hence, since $\mathcal{X}_{i}$ is a partial partition of $V$ and the
	depth of the \SPD is $k$, we get that $f(v)$ is nonzero in at most
	$3k$ coordinates for each $v\in V$.
\end{observation}

\section{The Analysis}
\label{sec:analysis}

The main technical lemmas now show that the per-coordinate expansion is
constant, and that for every pair, there exists a coordinate for which
the {\em expected} contraction is constant.
\begin{restatable}[Expansion Bound]{lemma}{ExpLem}
	\label{lem:expansion}
	For any vertices $u,v$, every coordinate $j$, and every choice
	of $\alpha,\beta$,
	\[ \left|f_j(v)-f_j(u)\right|=O(d_G(u,v)). \]
\end{restatable}

\begin{restatable}[Contraction Bound]{lemma}{ConLem}
	\label{lem:contraction}
	For any  vertices $u,v$, there exists some coordinate $j$ such that
	\[ \mathbb{E}_{\alpha,\beta}\left[\left|f_j(v)-f_j(u)\right|\right]=\Omega(d_G(u,v)). \]
\end{restatable}

Given these two lemmas, we can combine them together to show that the
entire embedding has small distortion. (A proof of the composition lemma
can be found in \Cref{app:ProofOfProbEmbed}.)
\begin{restatable}[Composition Lemma]{lemma}{ProbEmbed}\label{lem:ProbEmbed}
	Let $(X,d)$ be a metric space. Suppose that there are parameters
	$\rho,\tau$ and a function $f:X\rightarrow \R^s$, drawn from
	some probability space such that:
	\begin{enumerate}
		\item For every $u,v\in X$ and every $j\in[s]$,
		$|f_j(v)-f_j(u)|\le \rho\cdot d(v,u)$.
		\label{Prop:expansion}
		\item For every $u,v\in X$, there exists $j\in[s]$ such that
		$\mathbb{E}[\,|f_j(v)-f_j(u)|\,]\ge \frac{1}{\tau}\cdot
		d(v,u)$.\label{Prop:contraction}
		\item For every $v\in X$, $\|f(x)\|_0\le k$, that is $f(v)$ has support of size at most $k$
		(formally, there is a subset of indices $I_v\subseteq
		[s]$ of size $\leq k$, such that $\forall j\notin I_{v}$, $f_j(v)=0$).
		\label{Prop:arity}
	\end{enumerate}
	Then, for every $p\ge 1$, there is an embedding of $(X,d)$ into
	$\ell_p$ with distortion
	$O(\rho^{3}\tau^{2})\cdot k^{\nicefrac{1}{p}}$. In particular, if $\rho$ and $\tau$ are constant, then the distortion is constant.
	Moreover, if there is an efficient algorithm for sampling such an $f$,
	then there is a randomized algorithm that constructs the embedding
	efficiently (in expectation).
\end{restatable}

\Cref{thm:main}, our embedding for graphs with low depth \SPDs,
immediately follows by applying the Composition Lemma
(\Cref{lem:ProbEmbed}) to \Cref{lem:expansion},
\Cref{lem:contraction}, and \Cref{obv:coordinates}.

\subsection{Bounding the Expansion: Proof of \Cref{lem:expansion}}\label{Appendix:expansion}

In this section we bound the expansion of any coordinate in our
embedding. Recall that the embedding of $v$ lying in some component $X$
consists of two sets of coordinates: its distance from the path, and its
distance from the root. As mentioned in the introduction, since points
outside $X$ are mapped to zero, maintaining Lipschitz-ness requires us
to truncate the contribution of $v$ of any coordinate to its distance
from the boundary. This truncation (either via taking a minimum with
$d_G(v, V \setminus X)$, or via the sawtooth function), means that our
proofs of expansion require more care. Let us now give the details.

Consider any level $i$, any set $X\in\mathcal{X}_i$, and any pair of
vertices $u,v$. It suffices to show that $\Vert
f_{X}(v)-f_{X}(u)\Vert_{\infty} = O(d_G(u,v))$. To begin, we may assume
that both $u,v \in X$. Indeed, if both $u,v\notin X$, then
$f_{X}(v)=f_{X}(u)=\vec{0}$ and we are done. If one of them, say $v$
belongs to $X$ while the other $u\notin X$, then $f_{X}(u)=\vec{0}$
while $f_{X}(v)$ is bounded by $2^{t_{v}+1}\le 4 d_{G}\left(v,V\setminus
X\right) \le 4\, d_G(u,v)$ in each coordinate.

Moreover, we may also assume that the shortest $u$-$v$ path in $G$
contains only vertices from $X$. Indeed, suppose their shortest path in
$G$ uses vertices from $V\setminus X$, then $d_{G}(u,V\setminus
X)+d_{G}(v,V\setminus X)\le d_{G}(u,v)$. But since both
$f_{X}(v),f_{X}(u)$ are bounded in each coordinate by $4\cdot\max\left\{
d_{G}(u,V\setminus X),d_{G}(v,V\setminus X)\right\}$, we
have constant expansion. Henceforth, we can assume that
$d_{G}(u,v)=d_{X}(u,v)$. We now bound the expansion in each of the two
parts of $f_{X}$ separately.

\medskip \textbf{Expansion of ${f}^{\fpath}_{X}$.}  Let $X_{v},X_{u}$ be
the connected components in $G\left[X\setminus P_{X}\right]$ such that
$v\in X_{v}$ and $u\in X_{u}$.  Consider the first case $X_{v}\ne
X_{u}$, then $P_{X}$ intersects the shortest path between $v$ and
$u$. In particular,
\begin{align*}
\Vert{f}_{X}^{\fpath}(v)-{f}_{X}^{\fpath}(u)\Vert_{\infty} & \le\,\,\min\left\{ d_{X}\left(v,P_{X}\right),2d_{G}\left(v,V\setminus X\right)\right\} +\min\left\{ d_{X}\left(u,P_{X}\right),2d_{G}\left(u,V\setminus X\right)\right\} \\
& \le\,\,d_{X}\left(v,P_{X}\right)+d_{X}\left(u,P_{X}\right)\le d_{X}(v,u)~=~d_{G}(v,u)~.
\end{align*}
Otherwise, $X_{v}=X_{u}$ and the two vertices lie in the same
component. Now
$\Vert{f}^{\fpath}_{X}(v)-{f}^{\fpath}_{X}(u)\Vert_{\infty}$ equals
$
\big|\min\left\{ d_{X}\left(v,P_{X}\right),2 d_{G}\left(v,V\setminus
X\right)\right\} -\min\left\{ d_{X}\left(u,P_{X}\right),2
d_{G}\left(u,V\setminus X\right)\right\}\big|
$.
Assuming (without loss of generality) that the first term is at least the second, we can drop the absolute
value signs. Now the bound on the expansion follows from a simple case
analysis. Indeed, suppose $d_{X}\left(u,P_{X}\right) \leq 2
d_{G}\left(u,V\setminus X\right)$. Then we get
\begin{align*}
\Vert{f}_{X}^{\fpath}(v)-{f}_{X}^{\fpath}(u)\Vert_{\infty} & =\,\,\min\left\{ d_{X}\left(v,P_{X}\right),2d_{G}\left(v,V\setminus X\right)\right\} -d_{X}\left(u,P_{X}\right)\\
& \leq\,\,d_{X}\left(v,P_{X}\right)-d_{X}\left(u,P_{X}\right)~\leq~ d_{X}(u,v)~=~d_{G}(u,v).
\end{align*}
The other case is that $d_{X}\left(u,P_{X}\right) > 2 d_{G}\left(u,V\setminus
X\right)$, and then
\begin{align*}
\Vert{f}_{X}^{\fpath}(v)-{f}_{X}^{\fpath}(u)\Vert_{\infty} & =\,\,\min\left\{ d_{X}\left(v,P_{X}\right),2d_{G}\left(v,V\setminus X\right)\right\} -2d_{G}\left(u,V\setminus X\right)\\
& \leq\,\,2d_{G}\left(v,V\setminus X\right)-2d_{G}\left(u,V\setminus X\right)~\leq~2d_{G}(u,v).
\end{align*}
Hence the expansion is bounded by $2$.

\medskip \textbf{Expansion of ${f}^{\froot}_{X}$.} Let $r$ be the root
of $P_{X}$. For $t\in\{0,1,\dots,M\}$, let $p_{t}$ (respectively,
$q_{t}$) be the ``weight'' of $v$ (respectively, $u$) on $g_{t, \alpha, \beta}$---in
other words, $p_{t}$ is the constant in~(\ref{eq:ftilde}) such that
${f}^{\froot}_{X,t}(v)=p_{t}\cdot g_{t, \alpha, \beta}(d_X(v,r))$. Note that $p_{t} \in \{0, \lambda_v, 1 -
\lambda_v\}$ is chosen deterministically, and is nonzero for at most
two indices $t$.

First, observe that for every~$t$,
\begin{align}
\left|{f}_{X,t}^{\froot}(v)-{f}_{X,t}^{\froot}(u)\right|\notag
& =\left|p_{t}\cdot g_{t,\alpha,\beta}(d_{X}(v,r))-q_{t}\cdot g_{t,\alpha,\beta}(d_{X}(u,r))\right|\notag\\
& \le\min\left\{ p_{t},q_{t}\right\} \cdot\left| g_{t,\alpha,\beta}(d_{X}(v,r))-g_{t,\alpha,\beta}(d_{X}(u,r))\right|+\left|p_{t}-q_{t}\right|\cdot2^{t}\notag\\
& \le\min\left\{ p_{t},q_{t}\right\} \cdot\beta\cdot\left|d_{X}(v,r)-d_{X}(u,r)\right|+\left|p_{t}-q_{t}\right|\cdot2^{t}\notag\\
& \le O(d_{G}(u,v))+\left|p_{t}-q_{t}\right|\cdot2^{t}~.\label{eq:2}
\end{align}

The first inequality used that $g_{t, \alpha, \beta}$ is bounded by $2^t$, and the second
inequality that $g_{t, \alpha, \beta}$ is $\beta$-Lipschitz; both follow from
\obsref{Obs:Sawtooth}. The last inequality follows by the triangle
inequality (since we assumed that the shortest path from $v$ to $u$
is contained within $X$).

Hence, it suffices to show that
$\left|p_{t}-q_{t}\right|=O(d_G(u,v)/2^{t})$. Indeed, for indices
$t\notin\left\{ t_{u},t_{u}+1,t_{v},t_{v}+1\right\} $, $p_{t}=q_{t}=0$,
hence $\left|p_{t}-q_{t}\right|=0$. Let us consider the other cases.
W.l.o.g., assume that $d_{G}\left(v,V\setminus X\right)\ge
d_{G}\left(u,V\setminus X\right)$ and hence $t_{v}\ge t_{u}$.
\begin{itemize}
	\item $\boldsymbol{t_{u}=t_{v}:}$ In this case,
	$\left|p_{t_{v}}-q_{t_{v}}\right|=\left|(1-\lambda_{v})-(1-\lambda_{u})\right|
	= \lambda_{v}-\lambda_{u} =
	\left|p_{t_{v}+1}-q_{t_{v}+1}\right|$. Moreover, this quantity is
	\begin{align*}
	\lambda_{v}-\lambda_{u} & =~\frac{2d_{G}\left(v,V\setminus X\right)-2^{t}}{2^{t}}-\frac{2d_{G}\left(u,V\setminus X\right)-2^{t}}{2^{t}}\\
	& =~\frac{2\left(d_{G}\left(v,V\setminus X\right)-d_{G}\left(u,V\setminus X\right)\right)}{2^{t}}\\
	& \le~\frac{2d_{G}(u,v)}{2^{t}}~.
	\end{align*}
	Hence, we get that $\left|p_{t}-q_{t}\right|=O(d_G(u,v)/2^t)$ for
	both $t \in \{t_v, t_v+1\}$.
	
	\item $\boldsymbol{t_{u}=t_{v}-1}:$ It holds that	
	\begin{align*}
	\lambda_{v}+(1-\lambda_{u}) & \le\,\,2\cdot\frac{2d_{G}\left(v,V\setminus X\right)-2^{t_{v}}}{2^{t_{v}}}+\frac{2^{t_{u}+1}-2d_{G}\left(u,V\setminus X\right)}{2^{t_{u}}}\\
	& =\,\,\frac{2d_{G}\left(v,V\setminus X\right)-2d_{G}\left(u,V\setminus X\right)}{2^{t_{u}}}~\le~\frac{2d_{G}(u,v)}{2^{t_{u}}}~.
	\end{align*}
	If we define $\chi := \lambda_{v}+(1-\lambda_{u})$, we conclude that
	\begin{align*}
	\left|p_{t_{v}+1}-q_{t_{v}+1}\right|\qquad=\qquad\lambda_{v} & \,\,\le\,\,\chi\hspace{5pt}=\hspace{5pt}O(d_{G}(u,v)/2^{t_{v}+1})\\
	\left|p_{t_{v}}-q_{t_{v}}\right|\quad=\hspace{5pt}\left|1-\lambda_{v}-\lambda_{u}\right| & \,\,\le\,\,\chi\hspace{5pt}=\hspace{5pt}O(d_{G}(u,v)/2^{t_{v}})\\
	\left|p_{t_{u}}-q_{t_{u}}\right|\quad\quad=\hspace{11pt}\quad1-\lambda_{u} & \,\,\le\,\,\chi\hspace{5pt}=\hspace{5pt}O(d_{G}(u,v)/2^{t_{u}})
	\end{align*}
	
	\item $\boldsymbol{t_{u}<t_{v}-1}:$ By the definition of $t_v$ and
	$t_u$,
	$$2 d_{G}(v,u)\ge 2\left(d_{G}(v,V\setminus X)-d_{G}(u,V\setminus
	X)\right)\ge2^{t_{v}}-2^{t_{u}+1}\ge2^{t_{v}-1}~.$$ 
	In particular, 	for every $t\le t_v+1$,
	$\left|p_{t}-q_{t}\right|\le1\le\frac{2d_G(u,v)}{2^{t_{v}-1}} =
	O\left(\frac{d_G(u,v)}{2^{t}}\right)$.
\end{itemize}

\subsection{Bounding the Contraction: Proof of
	\Cref{lem:contraction}}
\label{sec:Con}

\newcommand{\distuv}{\Delta_{uv}}

Recall that we want to prove that for any pair $u,v$ of vertices, the
embedding has a large contribution between them.  A natural proof idea
is to show that vertices $u, v$ would eventually be separated by the
recursive procedure. When they are separated, either one of $u,v$ is far
from the separating path $P$, or they both lie close to the path. In the
former case, the distance $d(v,P)$ gives a large contribution to the
embedding distance, and in the latter case the distance from one end of
the path (the ``root'') gives a large contribution.

However, there's a catch: the value of $v$'s embedding in any single
coordinate cannot be more than $v$'s distance to the boundary, and this
causes problems. Indeed, if $u,v$ fall very close to the path $P$ at
some step of the algorithm, they must get most of their contribution at
this level, since future levels will not give much contribution. How can
we do it, without assigning large values?  This is where we use the
sawtooth function: it gives a good contribution between points without
assigning any vertex too large a value in any coordinate.

Formally, to bound the contraction and prove \Cref{lem:contraction},
for nodes $u, v$ we need to show that there exists a coordinate $j$ such
that $\E_{\alpha, \beta} [ |f_j(v)-f_j(u) | ]=\Omega(d_G(u,v))$.  For
brevity, define
\begin{gather}
\distuv := d_G(u,v).
\end{gather}

Fix $c=12$. Let $i$ be the minimal index such that there exists
$X\in\mathcal{X}_i$ with $u,v\in X$, and at least one of the following
holds:
\begin{enumerate}
	\item \label{item:pathClose} $\min\left\{ d_{X}(v,P_{X}),d_{X}(u,P_{X})\right\} \le\distuv/c$
	(i.e., we choose a path close to $\left\{ v,u\right\} $).
	\item \label{item:separate} $v$ and $u$ are in different components of $X\setminus P_{X}$.
\end{enumerate}
Note that such an index $i$ indeed exists: if $v$ and $u$ are separated by the \SPD then condition \conref{item:separate} holds. The only other possibility that $v$ and $u$ are never separated is when at least one of them lies on one of the shortest paths. In such a case, surely \conref{item:pathClose} holds.
By the
minimality of $i$, for every $X'\in \mathcal{X}_{i'}$ such that $i'<i$
and $u,v\in X'$, necessarily
$\min\{d_{X'}(v,P_{X'}),d_{X'}(u,P_{X'})\}>\distuv/c$. Therefore, the
ball with radius $\distuv/c$ around each of $v,u$ is contained in $X$. In
particular, $\min\left\{ d_{G}(v,V\setminus X),d_{G}(u,V\setminus
X)\right\} >\distuv/c$.

Suppose first that  \conref{item:separate} occurs but not \conref{item:pathClose}. Let $j$ be the coordinate in
${f}^{\fpath}_{X}$ created for the connected component of $v$ in
$X\setminus P_{X}$. Then
\begin{align*}
 \left|({f}_{X}^{\fpath})_{j}(v)-({f}_{X}^{\fpath})_{j}(u)\right|
& \quad=\quad\min\left\{ d_{X}\left(v,P_{X}\right),2d_{G}\left(v,V\setminus X\right)\right\} -0\\
& \quad\ge\quad\min\left\{ \frac{\distuv}{c},2\frac{\distuv}{c}\right\} ~=~\frac{\distuv}{c}~.
\end{align*}
Next assume that \conref{item:pathClose} occurs. W.l.o.g., assume that $d_{X}(v,P_{X})\le
d_{X}(u,P_{X})$, so that $d_X(v,P_X)\le\distuv/c$.  Suppose first that $d_{X}(u,P_{X})\ge2\distuv/c$. Then in the
coordinate $j$ in ${f}^{\fpath}_{X}$ created for the connected component
of $u$ in $X\setminus P_{X}$, we have
\begin{align*}
\left|({f}_{X}^{\fpath})_{j}(v)-({f}_{X}^{\fpath})_{j}(u)\right| & \ge\left|\min\left\{ d_{X}(u,P_{X}),2d_{G}(u,V\setminus X)\right\} -\min\left\{ d_{X}(v,P_{X}),2d_{G}(v,V\setminus X)\right\} \right|\\
& \ge\min\left\{ 2\frac{\distuv}{c},2\frac{\distuv}{c}\right\} -\frac{\distuv}{c}~=~\frac{\distuv}{c}~.
\end{align*}
(It does not matter whether $v$, $u$ are in the same connected component
or not.) Thus it remains to consider the case
$d_{X}(u,P_{X})<2\distuv/c$. Let $r$ be the root of $P_{X}$. Let $v'$
(resp. $u'$) be the closest vertex on $P_{X}$ to $v$ (resp. $u$) in
$G[X]$. Then by the triangle inequality
\begin{align*}
d_{X}(v',u') & \ge
d_{X}(v,u)-d_{X}(v,v')-d_{X}(u,u')\ge\frac{c-3}{c}\distuv~.
\end{align*}
In particular,
\begin{align}
\left|d_{X}(v,r)-d_{X}(u,r)\right| & \ge\,\,\left|d_{X}(v',r)-d_{X}(u',r)\right|-d_{X}(v,v')-d_{X}(u,u')\notag\\
& \ge\,\,\frac{c-6}{c}\distuv\,\,=\,\,\frac{1}{2}\distuv~,\label{inq:xyDist}
\end{align}
where we used that $P_{X}$ is a shortest path in $G[X]$
(implying $\left|d_{X}(v',r)-d_{X}(u',r)\right|=d_{X}(v',u')$).
See~\Cref{fig:figRoutevu} for illustration.

\begin{figure}[h]
	\centering{\includegraphics[scale=1]{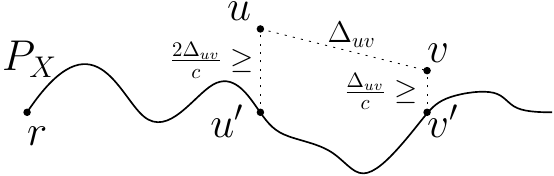}}
	\caption{\label{fig:figRoutevu}\small$P_X$ is a shortest path with
		root $r$. $v$ (resp. $u$) is at distance at most
		$\frac{\distuv}{c}$ (resp. $\frac{2\distuv}{c}$) from $v'$ (resp
		$u'$), it's closest vertex on $P_X$. By triangle inequality
		$d_X(v',u')\ge (1-\frac{3}{c})\distuv$. As $u',v'$ lay on the same
		shortest path starting at $r$,
		$|d_X(v',r)-d_X(u',r)|=d_X(v',u')$. Using the triangle inequality
		again we conclude
		 $\left|d_{X}(v,r)-d_{X}(u,r)\right|\ge\left|d_{X}(v',r)-d_{X}(u',r)\right|-\frac{3}{c}\distuv\ge(1-\frac{6}{c})\distuv$. }
\end{figure}

Set $x=d_{X}(v,r)$ and $y=d_{X}(u,r)$.  Assume first that
$d_{G}\left(v,V\setminus X\right)\ge d_{G}\left(u,V\setminus
X\right)$. In particular, $t_v\ge t_u$.  By the definition of $t_{v}$,
$2 d_{G}\left(v,V\setminus X\right)\le2^{t_{v}+1}$.  Thus
\begin{equation}\label{eq:2tv}
2^{t_v}\ge\frac{\distuv}{c}=\Omega(\distuv)
\end{equation}

\begin{claim}
	\label{clm:ExpLowerBound}
	Let $t\ge t_v$, then there is a constant $\phi$ such that
	\[ \mathbb{E}_{\alpha,\beta}\left[\left|g_{t, \alpha, \beta}(
	x)-g_{t, \alpha, \beta}(
	y)\right|\right]\ge \distuv/\phi. \]
\end{claim}
\begin{proof}
	According to \propertyref{Pr:SmallandLargeGap} of \Cref{lem:BrkIdn}
	$$\mathbb{E}_{\alpha,\beta}\left[\,\left|g_{t, \alpha, \beta}(
	x)-g_{t, \alpha, \beta}(y)\right|\,\right]= \Omega(\min\{|x-y|,
	2^t\})\overset{(\ref{inq:xyDist})\&(\ref{eq:2tv})}{=}\Omega(\distuv)~.$$		
\end{proof}
Set $S=\max\left\{ 8\phi,\frac{8c}{2}\right\}$.
Note that
$p_{t_{v}}+p_{t_{v}+1}=(1-\lambda_{v})+\lambda_{v}=1$.  Let
$t\in\{t_{v},t_{v}+1\}$ be such that $p_{t}\ge\frac{1}{2}$.
We consider two cases:
\begin{itemize}
	\item If $\left|p_{t}-q_{t}\right|\cdot2^{t}>\frac{\distuv}{S}$, then
	\begin{align}
\mathbb{E}_{\alpha,\beta}\left[\left|{f}_{X,t}^{\froot}(v)-{f}_{X,t}^{\froot}(u)\right|\right] & =\,\,\mathbb{E}_{\alpha,\beta}\left[\left|p_{t}\cdot g_{t,\alpha,\beta}(x)-q_{t}\cdot g_{t,\alpha,\beta}(y)\right|\right]\nonumber\\
& \ge\,\,\left|p_{t}\cdot\mathbb{E}_{\alpha,\beta}\left[g_{t,\alpha,\beta}(x)\right]-q_{t}\cdot\mathbb{E}_{\alpha,\beta}\left[g_{t,\alpha,\beta}(y)\right]\right|\nonumber\\
& =\,\,\left|p_{t}-q_{t}\right|\cdot2^{t}=\Omega(\distuv)~.\label{eq:ptqtLarge}
	\end{align}
	
	Where the equality follows by \propertyref{Pr:Avg} of
	\Cref{lem:BrkIdn}.
	
	\item Otherwise, using inequality
	(\ref{eq:2tv}), $q_{t}\ge p_{t}-\frac{\distuv}{2^{t_{v}}\cdot
		S}\ge\frac{1}{2}-\frac{2c}{2\distuv}\cdot\frac{\distuv}{S}\ge\frac{1}{4}$.
	In particular,
\begin{align}
\mathbb{E}_{\alpha,\beta}\left[\left|{f}_{X,t}^{\froot}(v)-{f}_{X,t}^{\froot}(u)\right|\right] & =\,\,\mathbb{E}_{\alpha,\beta}\left[\left|p_{t}\cdot g_{t,\alpha,\beta}(x)-q_{t}\cdot g_{t,\alpha,\beta}(y)\right|\right]\nonumber\\
& \ge\,\,\min\left\{ p_{t},q_{t}\right\} \cdot\mathbb{E}_{\alpha,\beta}\Big[\Big|g_{t,\alpha,\beta}(x)-g_{t,\alpha,\beta}(y)\Big|\Big]-\left|p_{t}-q_{t}\right|\cdot2^{t}\nonumber\\
& \ge\,\,\frac{1}{4}\cdot\frac{\distuv}{\phi}-\frac{\distuv}{S}=\Omega(\distuv)~,\label{eq:ptqtSmall}
\end{align}
	where in the first inequality we used \propertyref{Pr:Avg} of
	\Cref{lem:BrkIdn}, and in the second inequality we used
	\Cref{clm:ExpLowerBound}.
\end{itemize}

Finally, recall that we assumed $d_{G}\left(v,V\setminus X\right)\ge
d_{G}\left(u,V\setminus X\right)$ for the proof above. The other case
($d_{G}\left(v,V\setminus X\right)< d_{G}\left(u,V\setminus X\right)$)
is completely symmetric.

\section{The Composition Lemma: Proof of \Cref{lem:ProbEmbed}}\label{app:ProofOfProbEmbed}
We restate the lemma for convenience:
\ProbEmbed*
\begin{proof}
	
	Fix $n=|X|$, and
	set $m=48\rho\tau\cdot \ln n$. Let
	$f^{(1)},f^{(2)},\dots,f^{(m)}:X\rightarrow\R^s$ be functions chosen i.i.d according to the given distribution.
	Set $g=m^{-\nicefrac{1}{p}}\bigoplus_{i=1}^{m}f^{(i)}$.
	We argue that with high probability, $g$ has distortion $16\rho^{3}\tau^{2}\cdot k^{\nicefrac{1}{p}}$ in $\ell_p$.
	
	Fix some pair of vertices $v,u\in V$. Set $d(v,u)=\Delta$.
	The upper bound follows from \propertyref{Prop:expansion} and \propertyref{Prop:arity} of the lemma:
	\begin{align*}
	\left\Vert g(v)-g(u)\right\Vert _{p}^{p} & =\sum_{i=1}^{m}\sum_{j\in I_{v}\cup I_{u}}\left(m^{-\nicefrac{1}{p}}\cdot\left|f^{(i)}_{j}(v)-f^{(i)}_{j}(u)\right|\right)^{p}\\
	& \le\sum_{i=1}^{m}\sum_{j\in I_{v}\cup I_{u}}\frac{1}{m}\cdot(\rho\cdot\Delta)^{p}\le 2k\cdot(\rho\cdot\Delta)^{p}~,
	\end{align*}
	thus $\left\Vert g(v)-g(u)\right\Vert _{p}\le2^{\frac{1}{p}}\cdot\rho\cdot k^{\nicefrac{1}{p}}\cdot\Delta$.
	
	Next, for the contraction bound, let $j$ be the index of \propertyref{Prop:contraction} w.r.t. $v,u$.
	Set
	$\mathcal{F}=\{f~:~\left|f_{j}(v)-f_{j}(u)\right|\ge\Delta/2\tau\}$ to be the event that we draw a function with significant contribution to $v,u$.
	Then using \propertyref{Prop:expansion} and \propertyref{Prop:contraction},
	\begin{align*}
	\frac{\Delta}{\tau} & \le\mathbb{E}\left[\left|f_{j}(v)-f_{j}(u)\right|\right]\\
	& \le\Pr\left[\overline{\mathcal{F}}\right]\cdot\frac{\Delta}{2\tau}+\Pr\left[\mathcal{F}\right]\cdot\rho\Delta\le\frac{\Delta}{2\tau}+\Pr\left[\mathcal{F}\right]\cdot\rho\Delta~,
	\end{align*}
	which implies that $\Pr\left[\mathcal{F}\right]\ge\frac{1}{2\rho \tau}$.
	Let $Q^{(i)}_{u,v}$ be an indicator random variable for the event $f^{(i)}\in\mathcal{F}$, and set $Q_{u,v}=\sum_{i=1}^m Q^{(i)}_{u,v}$. By linearity of expectation, $\mathbb{E}[Q_{u,v}]\ge \frac{m}{2\rho \tau}=24\cdot\ln n$.
	By a Chernoff bound
	\begin{align*}
	\Pr\left[Q_{u,v}\le12\cdot\ln n\right] & \le\Pr\left[Q_{u,v}\le\frac{1}{2}\cdot\mathbb{E}\left[Q_{u,v}\right]\right]\\
	& \le\exp(-\frac{1}{8}\mathbb{E}\left[Q_{u,v}\right])\\
	& \le\exp(-3\ln n)=n^{-3}~.
	\end{align*}
	By taking a union bound over the ${n\choose 2}$ pairs, with probability at least $1-\frac{1}{n}$, for every $u,v\in V$, $Q_{u,v}>12\ln n=\frac{m}{4\rho\tau}$. 
	If this event indeed occurs, then the contraction is indeed bounded:
	\begin{align*}
		\left\Vert g(v)-g(u)\right\Vert _{p}^{p} & \ge\sum_{i=1}^{m}\left(m^{-\nicefrac{1}{p}}\cdot\left|f_{j}^{(i)}(v)-f_{j}^{(i)}(u)\right|\right)^{p}\\
		& \ge\frac{1}{m}\sum_{i:Q_{u,v}^{(i)}=1}\left|f_{j}^{(i)}(v)-f_{j}^{(i)}(u)\right|^{p}\\
		& \ge\frac{Q_{u,v}}{m}\cdot\left(\frac{\Delta}{2\rho\tau}\right)^{p}=\frac{1}{4\rho\tau}\cdot\left(\frac{\Delta}{2\rho\tau}\right)^{p}~.
	\end{align*}
	In particular, for every $u,v$, $\left\Vert g(v)-g(u)\right\Vert _{p}\ge(\frac{1}{4\rho\tau})^{\frac{1}{p}}\cdot\frac{\Delta}{2\rho\tau}$.
	Combining the upper and lower bounds, we conclude that $g$ has distortion $2^{\frac{1}{p}}\cdot\rho\cdot k^{\nicefrac{1}{p}}\cdot2\rho\tau\cdot(4\rho\tau)^{\frac{1}{p}}=2^{1+\frac{3}{p}}\rho^{2+\frac{1}{p}}\tau^{1+\frac{1}{p}}\cdot k^{\nicefrac{1}{p}}$.
	
\end{proof}

\section{The Sawtooth Lemma: Proof of \Cref{lem:BrkIdn}}
\label{app:MissBkIdn}
\begin{figure}[h]
	\centering{\includegraphics[width=1.0\textwidth]{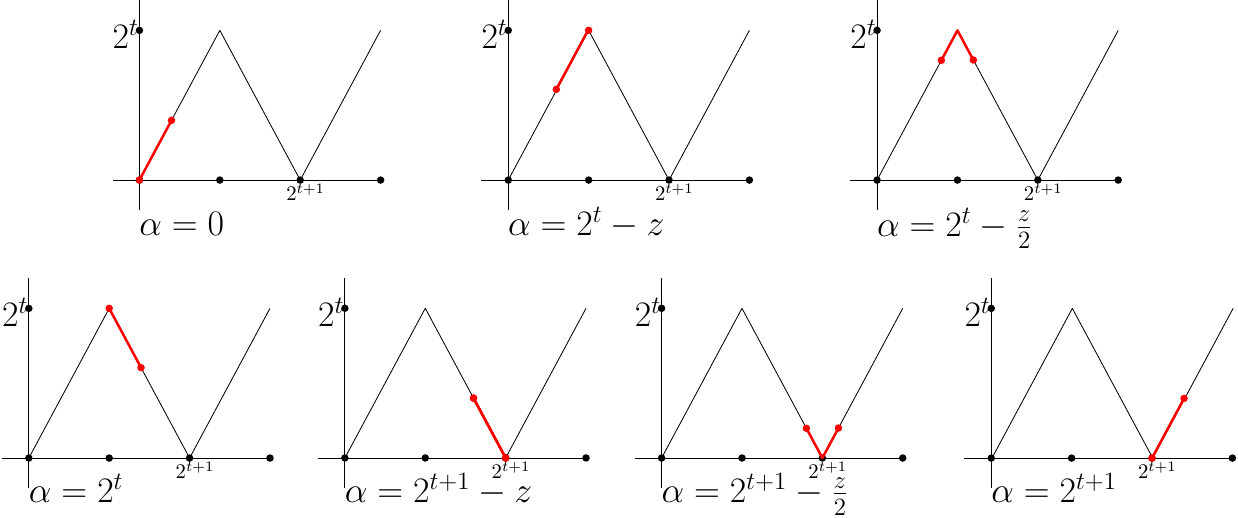}}
	\caption{\label{fig:PhaseChanges}\small{$\alpha$ is going from $0$ to $2^{t+1}$. $z\le 2^t$. In each of the figures the leftmost red point represents $\alpha$ while the rightmost red point represents $z+\alpha$. Each of the middle figures represent a moment when $g_{t}(z+\alpha)-g_{t}(z)$ changes its derivative.}}
\end{figure}

We restate \Cref{lem:BrkIdn} for convenience:
\BrkIdn*

\propertyref{Pr:Avg} is straightforward, as by \obsref{Obs:Sawtooth} $g_{t}$ is periodic with period length $2^{t+1}$. Indeed, for every fixed $\beta$, $\mathbb{E}_{\alpha}\left[g_{t, \alpha, \beta}(x)\right]=\mathbb{E}_{\alpha}\left[g_{t}(\beta x+\alpha\cdot2^{t+1})\right]=\mathbb{E}_{\alpha}\left[g_{t}(\alpha\cdot2^{t+1})\right]=2^{t-1}$.
The following claim will be useful in the proof of \propertyref{Pr:SmallandLargeGap}.
\begin{claim}
	\label{clm:BrIdnZ}For $z\in[0,2^{t+1}]$, $\mathbb{E}_{\alpha\in\left[0,1\right]}\left[\left|g_{t}(z+\alpha\cdot2^{t+1})-g_{t}(\alpha\cdot2^{t+1})\right|\right]=\frac{\left(2^{t+1}-z\right)z}{2^{t+1}}$.
\end{claim}
\begin{proof}
	Set $(*)=\mathbb{E}_{\alpha\in\left[0,1\right]}\left[\left|g_{t}(z+\alpha\cdot2^{t+1})-g_{t}(\alpha\cdot2^{t+1})\right|\right]$.
	By substituting the variable of integration,
	$(*)=\frac{1}{2^{t+1}}\cdot\int_{0}^{2^{t+1}}\left|g_{t}(z+\alpha)-g_{t}(\alpha)\right|d\alpha$.	
	First assume that $z\le2^{t}$, then there are $5$ ``phase changes''
	in $\left|g_{t}(z+\alpha)-g_{t}(\alpha)\right|$ from $0$ to $2^{t+1}$,
	at
	$2^{t}-z,~2^{t}-\frac{z}{2},~2^{t},~2^{t+1}-z,~2^{t+1}-\frac{z}{2}$.
	(see \Cref{fig:PhaseChanges} for illustration).
	
	We calculate
	\begin{align*}
	2^{t+1}\cdot(*) & =\int_{0}^{2^{t}-z}zd\alpha+\int_{0}^{\frac{z}{2}}(z-2\alpha)d\alpha+\int_{0}^{\frac{z}{2}}2\alpha d\alpha+\int_{0}^{2^{t}-z}zd\alpha+\int_{0}^{\frac{z}{2}}(z-2\alpha)d\alpha+\int_{0}^{\frac{z}{2}}2\alpha d\alpha\\
	& =2\cdot\int_{0}^{2^{t}-z}zd\alpha+2\cdot\int_{0}^{\frac{z}{2}}zd\alpha=\left(2^{t+1}-z\right)z\,.
	\end{align*}
	For $z>2^{t}$, set $w=2^{t+1}-z$. Then using that $g^{t}$ is periodic,
	\begin{align*}
	 \mathbb{E}_{\alpha\in\left[0,1\right]}\left[\left|g_{t}(w+\alpha\cdot2^{t+1})-g_{t}(\alpha\cdot2^{t+1})\right|\right]
	& =\,\,\mathbb{E}_{\alpha\in\left[0,1\right]}\left[\left|g_{t}(w+z+\alpha\cdot2^{t+1})-g_{t}(z+\alpha\cdot2^{t+1})\right|\right]\\
	& =\,\,\mathbb{E}_{\alpha\in\left[0,1\right]}\left[\left|g_{t}(2^{t+1}+\alpha\cdot2^{t+1})-g_{t}(z+\alpha\cdot2^{t+1})\right|\right]\\
	& =\,\,\mathbb{E}_{\alpha\in\left[0,1\right]}\left[\left|g_{t}(z+\alpha\cdot2^{t+1})-g_{t}(\alpha\cdot2^{t+1})\right|\right]~.
	\end{align*}
	Hence by the first case, $(*)=\frac{\left(2^{t+1}-w\right)w}{2^{t+1}}=\frac{\left(2^{t+1}-z\right)z}{2^{t+1}}$.
\end{proof}

For the proof of \propertyref{Pr:SmallandLargeGap} assume w.l.o.g.\ that $x>y$. Set $z=x-y$, and  $(*)=\mathbb{E}_{\alpha,\beta}\left[\left|g_{t, \alpha, \beta}(x)-g_{t, \alpha, \beta}(y)\right|\right]$.
As $g_{t}$ is a periodic function, we have that $(*)=\mathbb{E}_{\beta}\left[\mathbb{E}_{\alpha}\left[\left|g_{t, \alpha, \beta}(z)-g_{t, \alpha, \beta}(0)\right|\right]\right]$.
The rest of the proof is by case analysis.

\begin{itemize}
	\item \textbf{If} $\left|x-y\right|\le2^{t-1}$ \textbf{:}~~~
	Using \Cref{clm:BrIdnZ}, we have
	\begin{align*}
	& (*)=\mathbb{E}_{\beta}\left[\frac{\left(2^{t+1}-\beta z\right)\beta z}{2^{t+1}}\right]=\frac{1}{2^{t+1}}\cdot\frac{1}{4}\cdot\left(\frac{2^{t+1}z}{2}\beta^{2}-\frac{z^{2}}{3}\cdot\beta^{3}\mid_{0}^{4}\right)\\
	& ~~~~~=\frac{1}{4}\cdot\left(\frac{16}{2}\cdot z-\frac{z^{2}}{2^{t+1}}\cdot\frac{2^{6}}{3}\right)\ge\frac{1}{4}\cdot\left(8-\frac{64}{3\cdot4}\right)\cdot z=\frac{2}{3}\cdot|x-y|~,
	\end{align*}
	where in the inequality we used that $z\le2^{t-1}$.
	
	\item \textbf{If} $\left|x-y\right|>2^{t-1}$ \textbf{:}~~~
	As $g_{t}$ is periodic function,  \Cref{clm:BrIdnZ} implies that for every $w\ge 0$ it holds that $\mathbb{E}_{\alpha}\left[\left|g_{t, \alpha, \beta}(w)-g_{t, \alpha, \beta}(0)\right|\right]=\frac{\left(2^{t+1}-(w\mod2^{t+1})\right)(w\mod2^{t+1})}{2^{t+1}}$.
	Let $a\in\left[0,4\right]$ such that $a\cdot z=2^{t+1}$ (such $a$ exists as $\left|x-y\right|>2^{t-1}$).
	The claim follows as,
	\begin{align*}
	(*)\cdot2^{t+1} & =\,\,\mathbb{E}_{\beta\in[0,4]}\left[\left(2^{t+1}-(\beta z\mod2^{t+1})\right)(\beta z\mod2^{t+1})\right]\\
	& \ge\,\,\sum_{i=0}^{\left\lfloor \frac{4}{a}\right\rfloor -1}\frac{1}{4}\cdot\intop_{ia}^{(i+1)a}\left(2^{t+1}-(\beta\cdot\frac{2^{t+1}}{a}\mod2^{t+1})\right)\cdot(\beta\cdot\frac{2^{t+1}}{a}\mod2^{t+1})d\beta\\
	& =\,\,\sum_{i=0}^{\left\lfloor \frac{4}{a}\right\rfloor -1}\frac{1}{4}\cdot\intop_{0}^{a}\left(2^{t+1}-\beta\cdot\frac{2^{t+1}}{a}\right)\cdot\beta\cdot\frac{2^{t+1}}{a}d\beta\\
	& =\,\,\left\lfloor \frac{4}{a}\right\rfloor \cdot\frac{1}{4}\cdot\frac{a}{2^{t+1}}\cdot\intop_{0}^{2^{t+1}}\left(2^{t+1}-\gamma\right)\cdot\gamma d\gamma\\
	& \ge\,\,\frac{4}{2a}\cdot\frac{1}{4}\cdot\frac{a}{2^{t+1}}\cdot\left(2^{t+1}\frac{\gamma^{2}}{2}-\frac{\gamma^{3}}{3}\mid_{0}^{2^{t+1}}\right)\\
	& =\,\,\frac{1}{2}\cdot\frac{1}{2^{t+1}}\cdot\frac{\left(2^{t+1}\right)^{3}}{6}=\frac{\left(2^{t+1}\right)^{2}}{12}~.
	\end{align*}
\end{itemize}

\propertyref{Pr:SmallandLargeGap} now follows.

\section{Reducing the Dimension}\label{sec:dimension}
In the previous sections we did not attempt to bound the dimension of our embedding (\Cref{thm:main}). As each point is non-zero in at most $O(k\log n)$ coordinates (taking into account the repetitions done by the Composition Lemma (\Cref{lem:ProbEmbed})), naively we can bound the number of coordinates by $O(nk\log n)$. However, we can improve further. By introducing some modifications to the embedding algorithm, we are able to bound the number of coordinates by $O(k\log n)$.
Notice that this fact is interesting only for $p>2$. For embeddings into $\ell_2$, one can easily reduce the dimension to $O(\log n)$ using the Johnson Lindenstrauss lemma \cite{JL84}. Furthermore, for embeddings into $\ell_p$ for $p\in[1,2)$, we first embed into $\ell_2$ (using dimension $O(\log n)$). Then we embed from $\ell_2$ into $\ell_p$. It is well known that $\ell_2^d$ embeds into $\ell_p^{O(d)}$ (for $p\in[1,2]$) with constant distortion (see \cite{mat13lecture}),  thus we  conclude that our embedding can use only $O(\log n)$ coordinates.

\begin{theorem}[Embeddings with bounded dimension]
	\label{thm:dimension}
	Let $G=\left(V,E\right)$ be an $n$-vertex weighted graph with
	an \SPD of depth $k$. Then there exists an embedding
	$f:V\to\ell^{O(k\log n)}_{p}$ with distortion $O(k^{\nicefrac{1}{p}})$.
\end{theorem}

\begin{proof}
	Recall the embedding algorithm: we assumed that the minimal distance in $G$ is $1$, while the diameter is bounded by $2^M$. Let
	$\left\{\mathcal{X},\mathcal{P}\right\} =\left\{ \left\{ \mathcal{X}_{1},,\dots,\mathcal{X}_{k}\right\} ,\left\{ \mathcal{P}_{1},\dots,\mathcal{P}_{k}\right\} \right\} $ be an \SPD of depth
	$k$ for $G$.
	For every index $i\in[k]$ and cluster $X\in\mathcal{X}_i$ we had two different embeddings ${f}_{X}^{\fpath}$ and ${f}_{X}^{\froot}$.
	The function ${f}_{X}^{\fpath}$ is a deterministic embedding that maps each point $x\in X$ to its (truncated) distance from $P_X$, while using a different coordinate for each connected component in $X\setminus P_X$.
The function ${f}_{X}^{\froot}$ is an embedding that depends on random variables $\alpha,\beta$. It uses $M+1$ different coordinates that captures a randomly truncated distance to the root $r$ of $P_X$.
	
	In \Cref{lem:expansion} we proved that our embedding is Lipschitz in each coordinate. In \Cref{lem:contraction} we showed that for every pair of vertices $v,u$ there is some coordinate $j$ such that $\mathbb{E}\left[\left|f_j(v)-f_j(u)\right|\right]=\Omega(d_G(u,v))$. The coordinate $j$ might come from either ${f}^{\fpath}$ or ${f}^{\froot}$.
	Denote by $\mathcal{R}_\fpath\subseteq{V\choose2}$ (resp. $\mathcal{R}_\froot$) the set of pairs for which the coordinate above come from  ${f}^{\fpath}$ (resp.  ${f}^{\froot}$).
	In order to replace expectation with high probability, we invoke $O(\log n)$ independent repetitions of our embedding (\Cref{lem:ProbEmbed}).
	We will modify each type of coordinates separately, arguing that a total of $O(k\log n)$ coordinates suffices.
	
	\paragraph{${f}^{\fpath}$:} We start with modifying the ${f}^{\fpath}$ type coordinates. First, note that as the value of this coordinates chosen deterministically, there is no reason to invoke the independent repetitions (\Cref{lem:ProbEmbed}).
	Next, consider a specific level $i\in[k]$. For every cluster $X\in\mathcal{X}_{i+1}$, let $\Pi(X)\in\mathcal{X}_i$ be the cluster such that $X\subseteq \Pi(X)$.
	Denote by ${f}^{\fpath}_i$ the concatenation of all $\left({f}^{\fpath}_{\Pi(X)}\right)_X$ for $X\in\mathcal{X}_{i+1}$, and by ${f}^{\fpath}$ the concatenation of all ${f}^{\fpath}_i$ for $i\in[k]$.
	Set $D=|\mathcal{X}_{i+1}|$, note that ${f}^{\fpath}_i$ has exactly $D$ coordinates, where each $v\in V$ is non-zero in at most one coordinate.
	For every $X\in\mathcal{X}_{i+1}$ pick a sequence $\alpha^X\in \{\pm1\}^{m}$, where $m=O(\log D)$, such that for every different $X,X'\in\mathcal{X}_{i+1}$ the number of coordinates where $\alpha^X$ and $\alpha^{X'}$ differ is at least $\frac m4$.\footnote{Such a set of sequences can be chosen greedily.}
	We define a new embedding $h^{\fpath}_i:V\rightarrow\mathbb{R}^m$, such that for every $v\in X\in\mathcal{X}_{i+1}$, $h^{\fpath}_i(v)=\frac{f^{\fpath}_{\Pi(X)}(v)}{m^{\nicefrac1p}}\left(\alpha^{X}_1,\dots,\alpha^X_m\right)$. For $v\in V$ that belongs to no cluster in $\mathcal{X}_{i+1}$, set $h^{\fpath}_i(v)=\vec{0}$. Consider $v,u\in V$. If $u,v$ are both belong to the same cluster $X$, then
	\begin{align*}
	\left\Vert h_{i}^{\fpath}(v)-h_{i}^{\fpath}(u)\right\Vert _{p}^{p} & =\sum_{i=1}^{m}\left|\alpha_{i}^{X}\cdot\left(\frac{f_{\Pi(X)}^{\fpath}(v)}{m^{\nicefrac{1}{p}}}-\frac{f_{\Pi(X)}^{\fpath}(u)}{m^{\nicefrac{1}{p}}}\right)\right|^{p}\\
	& =\left|f_{\Pi(X)}^{\fpath}(v)-f_{\Pi(X)}^{\fpath}(u)\right|^{p}=\left\Vert f_{i}^{\fpath}(v)-f_{i}^{\fpath}(u)\right\Vert _{p}^{p}
	\end{align*}
	On the other hand, if $v\in X_v$ and $u\in X_u$ belong to different clusters, it holds that
	\begin{align*}
	\left\Vert h_{i}^{\fpath}(v)-h_{i}^{\fpath}(u)\right\Vert _{p}^{p} & =\sum_{i=1}^{m}\frac{1}{m}\left|\alpha_{i}^{X_{v}}\cdot f_{\Pi(X_{v})}^{\fpath}(v)-\alpha_{i}^{X_{u}}\cdot f_{\Pi(X_{u})}^{\fpath}(u)\right|^{p}\\
	\left\Vert h_{i}^{\fpath}(v)-h_{i}^{\fpath}(u)\right\Vert _{p}^{p} & \le\sum_{i=1}^{m}\frac{1}{m}\left|f_{\Pi(X_{v})}^{\fpath}(v)+f_{\Pi(X_{u})}^{\fpath}(u)\right|^{p}\\
	& \le2^{p}\cdot\left(\left(f_{\Pi(X_{v})}^{\fpath}(v)\right)^{p}+\left(f_{\Pi(X_{u})}^{\fpath}(u)\right)^{p}\right)=2^{p}\cdot\left\Vert f_{i}^{\fpath}(v)-f_{i}^{\fpath}(u)\right\Vert _{p}^{p}\\
	\left\Vert h_{i}^{\fpath}(v)-h_{i}^{\fpath}(u)\right\Vert _{p}^{p} & \ge\frac{m}{4}\cdot\frac{1}{m}\left|f_{\Pi(X_{v})}^{\fpath}(v)+f_{\Pi(X_{u})}^{\fpath}(u)\right|^{p}\\
	& \ge\frac{1}{4}\cdot\frac{1}{2}\left(\left(f_{\Pi(X_{v})}^{\fpath}(v)\right)^{p}+\left(f_{\Pi(X_{u})}^{\fpath}(u)\right)^{p}\right)=\frac{1}{8}\cdot\left\Vert f_{i}^{\fpath}(v)-f_{i}^{\fpath}(u)\right\Vert _{p}^{p}
	\end{align*}
	Note that $h_{i}^{\fpath}$ has $m=O(\log D)\le O(\log n)$ coordinates.
	Denote by ${h}^{\fpath}$ the concatenation of all ${h}^{\fpath}_i$ for $i\in[k]$. Then ${h}^{\fpath}$ has at most $O(k\log n)$ coordinates, as desired. Moreover, for all $u,v\in V$ it holds that
	\[
	8^{-\frac{1}{p}}\cdot\left\Vert f^{\fpath}(v)-f^{\fpath}(u)\right\Vert _{p}\le\left\Vert h^{\fpath}(v)-h^{\fpath}(u)\right\Vert _{p}\le2\cdot\left\Vert f^{\fpath}(v)-f^{\fpath}(u)\right\Vert _{p}~.
	\]

	\paragraph{${f}^{\froot}$:}  next we modify the $f^{\froot}$ type coordinates. Consider level $i\in[k]$, and a cluster $X\in\mathcal{X}_{i}$. $f^{\froot}_X:V\rightarrow\mathbb{R}^{M+1}$ is a function that sends each vertex $v\notin X$ to $\vec{0}$, while each vertex $v\in X$ has a specific scale $t_v\in [0,M-1]$, such that $f^{\froot}_X(v)$ can be nonzero only in coordinates $t_v,t_{v+1}$.
	Set $h^{\froot}_X:V\rightarrow\mathbb{R}^2$ as a concatenation of $h^{\froot}_{X,\odd},h^{\froot}_{X,\even}$, where $h^{\froot}_{X,\odd}$ (resp. $h^{\froot}_{X,\even}$) is the sum of all the odd (resp. even) coordinates of $f^{\froot}_X$. That is $h_{X,\odd}^{\froot}=\sum_{t=0}^{\left\lfloor \nicefrac{M-1}{2}\right\rfloor }f_{X,2t+1}^{\froot}$
	and $h_{X,\even}^{\froot}=\sum_{t=0}^{\left\lfloor \nicefrac{M}{2}\right\rfloor }f_{X,2t}^{\froot}$.
	Next define $h_{i}^{\froot}=\sum_{X\in\mathcal{X}_{i}}h_{X}^{\froot}$ as the sum of all
	$h_{X}^{\froot}$ for $X\in\mathcal{X}_{i}$.	
	Denote by $f^{\froot}_i$ the sum of all $f^{\froot}_X$ for $X\in\mathcal{X}_i$, and by ${f}^{\froot}$ the concatenation of all $f^{\froot}_i$ for $i\in[k]$.
	It is clear that the expansion is not increased in $h_{i}^{\froot}$, as for every $v,u\in V$, using the triangle inequality
	\begin{align*}
	\left\Vert h_{i}^{\froot}(v)-h_{i}^{\froot}(u)\right\Vert _{p} & \le\sum_{X\in\mathcal{X}_{i}}\left\Vert h_{i,X}^{\froot}(v)-h_{i,X}^{\froot}(u)\right\Vert _{p}\\
	& \le\sum_{X\in\mathcal{X}_{i}}\sum_{t=0}^{M}\left\Vert f_{X,t}^{\froot}(v)-f_{X,t}^{\froot}(u)\right\Vert _{p}=\left\Vert f_{i}^{\froot}(v)-f_{i}^{\froot}(u)\right\Vert _{p}~.
	\end{align*}

	Arguing that the expected contraction property is maintained is more involved.
	Consider a pair of vertices $v,u\in V$. Following the arguments in \Cref{lem:contraction}, $i$ is the minimal index such that there exists $X\in\mathcal{X}_i$ with $u,v\in X$ such that either \conref{item:pathClose} or \conref{item:separate} hold.
	We can assume that \conref{item:pathClose} holds, and moreover,
        that $d_X(v,P_X),d_X(u,P_X)\le2\Delta_{uv}/c$ (as otherwise the
        coordinate that contributes to the contraction comes from
        $f^\fpath$ and we have nothing to prove here). In particular,
        inequality \eqref{inq:xyDist}, inequality \eqref{eq:2tv} and
        \Cref{clm:ExpLowerBound} hold. Recall that we assumed  $t_v\ge
        t_u$, and let $t\in\{t_v,t_v+1\}$ such that
        $p_t\ge\frac12$. W.l.o.g., assume that $t$ is odd. We proceed to the case analysis:
	\begin{itemize}
		\item If $\left|p_{t}-q_{t}\right|\cdot2^{t}>\frac{\distuv}{S}$ and $q_t\ne 0$, note that for every odd $t'\ne t$, $f_{X,t'}^{\froot}(v)=f_{X,t'}^{\froot}(u)=0$. Therefore,  following inequality \eqref{eq:ptqtLarge} 	
		\[
		 \mathbb{E}_{\alpha,\beta}\left[\left|h_{i,\odd}^{\froot}(v)-h_{i,\odd}^{\froot}(u)\right|\right]=\mathbb{E}_{\alpha,\beta}\left[\left|f_{X,t}^{\froot}(v)-f_{X,t}^{\froot}(u)\right|\right]=\Omega(\distuv)~.
		\]
		\item Otherwise, if $q_t= 0$ there might be a single odd scale $t'\le t-2$ such that $q_{t'}\ne 0$ (if $q_{t'}=0$ for all odd scales, then the analysis above holds). We have
		\begin{align*}
		\mathbb{E}_{\alpha,\beta}\left[\left|h_{i,\odd}^{\froot}(v)-h_{i,\odd}^{\froot}(u)\right|\right] & =\mathbb{E}_{\alpha,\beta}\left[\left|f_{X,t}^{\froot}(v)-f_{X,t'}^{\froot}(u)\right|\right]\\
		& =\mathbb{E}_{\alpha,\beta}\left[\left|p_{t}\cdot g_{t, \alpha, \beta}(x)-q_{t'}\cdot g_{t', \alpha, \beta}(y)\right|\right]\\
		& \ge\left|p_{t}\cdot\mathbb{E}_{\alpha,\beta}\left[g_{t, \alpha, \beta}(x)\right]-q_{t'}\cdot\mathbb{E}_{\alpha,\beta}\left[g_{t', \alpha, \beta}(y)\right]\right|\\
		& \ge p_{t}\cdot2^{t-1}-q_{t'}\cdot2^{t'-1}\ge\frac{1}{2}\cdot2^{t-1}-2^{t-3}=2^{t-3}=\Omega(\distuv)~,
		\end{align*}
		where the last equality follows by inequality \eqref{eq:2tv}.

		\item Otherwise, $\left|p_{t}-q_{t}\right|\cdot2^{t}\le\frac{\distuv}{S}$. Using inequality \eqref{eq:2tv}, $q_t\ge\frac14$ (and therefore $f_{X,t'}^{\froot}(v)=f_{X,t'}^{\froot}(u)=0$ for every odd $t'\ne t$). Following inequality \eqref{eq:ptqtSmall},
		\[
		 \mathbb{E}_{\alpha,\beta}\left[\left|h_{i,\odd}^{\froot}(v)-h_{i,\odd}^{\froot}(u)\right|\right]=\mathbb{E}_{\alpha,\beta}\left[\left|f_{X,t}^{\froot}(v)-f_{X,t}^{\froot}(u)\right|\right]=\Omega(\distuv)~.
		\]
	\end{itemize}
Define ${h}^{\froot}$ the concatenation of all $h^{\froot}_i$ for $i\in[k]$. ${h}^{\froot}$ has exactly $2k$ coordinates. We saw that $h$ is Lipschitz in every coordinate. Moreover, for every $\{u,v\}\in\mathcal{R}_\froot$, $\mathbb{E}_{\alpha,\beta}\left[\left|h_{i,\odd}^{\froot}(v)-h_{i,\odd}^{\froot}(u)\right|\right]=\Omega(\distuv)$.

Set $h$ to be the concatenation of ${h}^{\fpath}$ and ${h}^{\froot}$. We now invoke the composition lemma (\Cref{lem:ProbEmbed}) to construct an embedding with distortion $O(k^{\frac1p})$.
Recall that during the construction of \Cref{lem:ProbEmbed} we sample  and concatenate $O(\log n)$ independent copies of $h$ (normalized accordingly). As ${h}^{\fpath}$ is deterministic, it is enough to take only a single (non-normalized) copy of ${h}^{\fpath}$, and $O(\log n)$ (normalized) copies of  ${h}^{\froot}$. In particular, the total number of coordinates is $O(k\log n)+O(\log n)\cdot 2k=O(k\log n)$, as required.
\end{proof}

\section{Conclusions}

In this paper we introduced the notion of shortest path decompositions with low depth. We showed how these can be used to give embeddings into
$\ell_p$ spaces. Our techniques give optimal embeddings of bounded
pathwidth graphs into $\ell_2$, and also new embeddings for graphs with
bounded treewidth, planar, and excluded-minor families of
graphs. Our embedding for the family of graphs with SPD depth $k$ into
$\ell_p$ has an asymptotically matching lower
bound for every fixed $p>1$. 
Our techniques already have been useful for other embedding results, e.g., for embedding planar graphs with small face covers into $\ell_1$~\cite{Fil20faces}.
We hope that our techniques will find further applications.

Our work raises several open questions.
While our embeddings are tight for fixed $p >1$, can we improve
the bounds for $\ell_1$ embedding of graphs with bounded pathwidth, or more generally to graphs with bounded \SPDdepth?
Can we give better results for the \SPDdepth of $H$-minor-free graphs? Our
approach gives a $O(\sqrt{\log n})$-distortion embedding of planar
graphs into $\ell_1$, which is quite different from the previous known
results using padded decompositions: can a combination of these ideas be used to make
progress towards the planar graph embedding conjecture?

\subsection*{Acknowledgments}
Anupam Gupta is supported in part by NSF awards CCF-1536002, CCF-1540541, and CCF-1617790.
Ofer Neiman is supported in part by ISF grant 1817/17, and by BSF Grant 2015813.

A previous version of this paper contained a lower bound for embeddings of graphs with bounded \SPDdepth into $\ell_1$, based on the diamondfold graph \cite{LS11}. Our proof was wrong, and hence this statement is removed from the current version.

{\small
  \bibliographystyle{alphaurlinit}
  \bibliography{bib-extended,art}
}

\newpage	
\end{document}